\documentclass[11pt]{article}
\usepackage[margin=1in]{geometry}
\usepackage[utf8]{inputenc}
\usepackage{amsmath,amsthm,amsfonts}
\usepackage[T1]{fontenc}
\usepackage[colorlinks=true,linkcolor=black!40!blue,citecolor=black!40!blue,hyperfootnotes=true]{hyperref}
\usepackage[capitalise]{cleveref}
\usepackage{xspace}
\usepackage{color}
\usepackage[pdftex,dvipsnames]{xcolor}  
\usepackage{tikz,pgfplots}
\usepgfplotslibrary{fillbetween}
\usetikzlibrary{patterns, calc}
\pgfplotsset{compat=1.15}
\usepackage[ruled]{algorithm2e}
\usepackage{tabu,multirow}
\usepackage{enumitem}
\usepackage{mathtools}
\usepackage{pdfpages}
\usepackage{titlesec}
\usepackage{eso-pic}
\usepackage[numbers]{natbib}
\usepackage{booktabs}
\usepackage{scalerel}
\usepackage{url}

\usepackage{bm} 

\AtBeginDocument{%
	\DeclareFontShape{T1}{cmr}{m}{scit}{<->ssub*cmr/m/sc}{}%
}

\let\truehypersetup\hypersetup
\renewcommand\hypersetup[1]{}
\usepackage{bigfoot}
\let\hypersetup\truehypersetup

\usepackage[font=small]{caption} 
\usepackage{thmtools}
\usepackage{thm-restate}

\newtheorem{theorem}{Theorem}[section]

\newtheorem{lemma}[theorem]{Lemma}

\newtheorem{coro}[theorem]{Corollary}

\newtheorem{manuallemmainner}{Lemma}
\newenvironment{manuallemma}[1]{%
	\renewcommand\themanuallemmainner{#1}%
	\manuallemmainner
}{\endmanuallemmainner}

\newcommand{\mc}[1]{\ensuremath{\mathcal{#1}}\xspace}

\newcommand{\opt}[1]{\ensuremath{\textsc{Opt}_{#1}}\xspace}
\newcommand{\cmax}{\ensuremath{C_{\max}}}
\newcommand{\dx}{\mathop{}\!\mathrm{d}x}

\newcommand{\rhosand}{\ensuremath{\bar{\rho}(m)}\xspace} 
\newcommand{\rhosandm}[1]{\ensuremath{\bar{\rho}(#1)}\xspace} 
\newcommand{\rhosandid}{\ensuremath{\bar{\rho}_{\textup{\tiny 01}}}\xspace} 
\newcommand{\rhosandidm}[1]{\ensuremath{\bar{\rho}_{\textup{01}}(m)}\xspace} 

\newcommand{\rhoid}{\rhosandid}

\newcommand{\R}{\mathbb{R}}
\newcommand{\N}{\mathbb{N}}
\newcommand{\Z}{\mathbb{Z}}
	
\newcommand{\nrbags}{m}
\newcommand{\nrmach}{{m'}}
\newcommand{\optm}{\ensuremath{\textsc{Opt}_{m}}\xspace}

\newcommand{\optf}{\ensuremath{\textsc{Opt}_{m'}}\xspace}

\newcommand{\lptf}{\ensuremath{\text{LPT}_{m'}}\xspace}

\newcommand{\simplefold}{\emph{simple folding}\xspace}

\newcommand{\compfunc}{g}

\newcommand{\nrfail}{t}
\newcommand{\avg}{\lambda\xspace}

\newcommand{\machineload}{P}

\newcommand{\algoname}[1]{\textsc{#1}\xspace}

\newcommand{\lpt}{\algoname{LPT}}  
\newcommand{\oddalgo}{\algoname{BuildOdd}} 
\newcommand{\sandtobricks}{\algoname{SandForBricks}} 
\newcommand{\sandalg}{\algoname{Sand}} 
\newcommand{\sandalgid}{$\algoname{Sand}_{01}$} 

\newcommand{\srs}{speed-robust scheduling\xspace}

\newcommand{\floor}[1]{\ensuremath{\left\lfloor#1\right\rfloor}\xspace}
\newcommand{\ceil}[1]{\ensuremath{\left\lceil#1\right\rceil}\xspace}

\DeclareMathOperator*{\argmax}{arg\,max}

\newcommand{\Mod}[1]{\ (\mathrm{mod}\ #1)}

\title{Speed-Robust Scheduling}
\author{Franziska Eberle\thanks{London School of Economics and Political Sciences, Department of Mathematics, United Kingdom. Email: \texttt{franziska.eberle@posteo.de}}
	\and Ruben Hoeksma\thanks{University of Twente, Department of Applied Mathematics, The~Netherlands. \texttt{r.p.hoeksma@utwente.nl}}
	\and Nicole Megow\thanks{University of Bremen, Department of Mathematics and Computer Science, Germany. \texttt{\{nmegow,noelke\}@uni-bremen.de}}
	\and Lukas Nölke\footnotemark[4]
	\and Kevin Schewior\thanks{University of Southern Denmark, Department of Mathematics and Computer Science, Odense, Denmark. \texttt{kevs@sdu.dk}}
	\and Bertrand Simon\thanks{IN2P3 Computing Center, CNRS, Villeurbanne, France. \texttt{bertrand.simon@cc.in2p3.fr}}
}

\date{\today}
\begin{document}

\title{Speed-Robust Scheduling\thanks{Research was partially supported by Deutsche Forschungsgemeinschaft (DFG, German Research Foundation) under project 146371743 within TRR 89 Invasive Computing  and under contract ME 3825/1.}\\\Large Sand, Bricks, and Rocks
}

\maketitle

\begin{abstract}
The \srs problem is a two-stage problem where, given $m$ machines, jobs must be grouped into
	at most~$m$ bags while the processing speeds of the 
	machines
	are unknown. After the speeds are revealed, the grouped jobs must be assigned
	to the machines without being separated. To evaluate the performance of
	algorithms, we determine upper bounds on the worst-case ratio of the
	algorithm's makespan and the optimal makespan given full information. We refer
	to this ratio as the robustness factor. We give an algorithm with a robustness
	factor~$2-\frac1m$ for the most general setting and improve this to~$1.8$ for
	equal-size jobs. For the special case of infinitesimal jobs, we give an
	algorithm with an optimal robustness factor {equal to~$\frac e{e-1}\approx 1.58$}. The
	particular machine environment in which all machines have either speed~$0$
	or~$1$ was studied before by Stein and Zhong~(ACM Trans.\ Alg.\ 2020). For this setting, we
	provide an algorithm for scheduling infinitesimal jobs with an optimal robustness
	factor of~$\frac{1+\sqrt{2}}2\approx 1.207$. It lays the foundation for an
	algorithm matching the lower bound of~$\frac43$ for equal-size jobs.
\end{abstract}

\section{Introduction}

	Scheduling problems with incomplete knowledge of the input data have been
	studied extensively. There are different ways to model such uncertainty, the major
	frameworks being {\em online optimization}, where parts of the input are
	revealed incrementally, {\em stochastic optimization}, where parts of the input
	are modeled as random variables, and {\em robust optimization}, where
	uncertainty in the data is bounded. %
	Most scheduling research in this context assumes uncertainty about the job
	characteristics. Examples include online scheduling, where the job set is a priori unknown
	\cite{AlbersH17,PruhsST04}, mixed-criticality scheduling, where
	the processing time 
	comes from a given set \cite{BaruahBDLMMS12},
	stochastic scheduling, where the processing times are modeled as random
	variables with known distributions~\cite{Nino-Mora09,Pinedo2012-book}, robust scheduling, where
	the unknown processing times are within a given interval \cite{KouvelisYu97},
	two/multi-stage stochastic and robust scheduling with recourse, where the set of
	jobs that has to be scheduled 
	stems from a known superset and is revealed in stages \cite{ChenMRS21,ShmoysS07}, 
	and scheduling with explorable uncertainty,
	where the processing time of a job can potentially be reduced by testing the job at some extra cost~\cite{DuerrEMM20}.
	
	A lot less research addresses uncertainty about the machine
	environment, particularly, where the processing speeds of machines change in an
	unforeseeable manner. The majority of such research focuses on the special case of
	scheduling with unknown non-availability periods, that is, machines break down
	temporarily~\cite{AlbersS01,DiedrichJST09} or permanently~\cite{SteinZ20}.
	Arbitrarily changing machine speeds have been considered for scheduling on a
	single machine~\cite{EpsteinLMMMSS12}.
	
	Fluctuations in the processing speeds of machines are pervasive in real-world
	environments. For example, machines can be shared computing resources in 
	data centers, where a sudden increase of workload may cause a general slowdown
	or, for some users, the machine may become unavailable due to priority issues.
	As another example, machines that are production units may change their
	processing speed due to aging or, unexpectedly, break down completely. In any
	case,~(unforeseen) changes in the processing speed may have a drastic impact on
	the quality of a given schedule.
	
	In this paper, we are concerned with the question of how to design a partial
	schedule by committing to groups of jobs, to be scheduled on the same machine,
	before knowing the actual machine speeds. This question is motivated, for
	example, by MapReduce computations done in large data centers. 
	A MapReduce
	function typically groups workload before knowing the actual number or precise
	characteristics of the available computing resources~\cite{MapReduce}.
		
	We consider a two-stage robust scheduling problem in which we aim for a schedule
	of minimum makespan on multiple machines of unknown speeds.
	Given a set of~$n$ jobs and~$\nrbags$ machines, we ask
	for a partition of the jobs into~$\nrbags$ groups, we say {\em bags}, that have to be scheduled on the machines after their speeds are revealed without being split up.
	That is, in the second stage, when the machine speeds are known, a feasible schedule assigns all jobs in the same
	bag to the same machine.
	The goal is to minimize the second-stage makespan.
			
	More formally, we define the {\em \srs} problem as follows. We are given~$n$
	jobs with processing times~$p_j\geq 0$, for~$j\in\{1, \ldots, n\}$, and the
	number of machines,~$\nrbags \in \N$. Machines run in parallel but their speed
	is a priori unknown. 
	In the first stage, the task is to group jobs into at most~$\nrbags$ bags. In
	the second stage, the machine speeds~$s_i\geq 0$, for~$i\in\{1,\ldots,m\}$, are
	revealed. The time needed to execute job~$j$ on machine~$i$ is~$p_j/s_i$, if~$s_i>0$. 
	If a machine has speed~$s_i=0$, then it cannot process any job; we say the
	machine {\em fails}. Given the machine speeds, the second-stage task is to
	assign bags to the machines such that the makespan~$\cmax$ is minimized, where
	the makespan is the maximum sum of execution times of jobs assigned to the same
	machine. 

	Given a set of bags and machine speeds, the second-stage problem emerges as
	classical makespan minimization on related
	parallel machines. It is well-known that this problem can
	be solved arbitrarily close to optimality by polynomial-time approximation
	schemes~\cite{AlonAWY1998,HochbaumS87,Jansen10}. As we are interested in the
	information-theoretic tractability, we
	allow superpolynomial running times for our algorithms -- ignoring any computational concern -- and assume that 
	the second-stage problem is solved optimally. Thus, an \emph{algorithm} for \srs defines a job-to-bag allocation, i.e., it gives a solution to the first-stage problem. We may use non-optimal bag-to-machine assignments to simplify the analysis.
	
	We evaluate the performance of algorithms by a worst-case analysis comparing the
	makespan of the algorithm with the optimal makespan achievable when all machine
	speeds are known in advance. We say that an algorithm is $\rho$-{\em robust} if,
	for any instance, its makespan is within a factor~$\rho\geq 1$ of the optimal
	solution. The {\em robustness factor} of the algorithm is defined as the infimum over all such~$\rho$.
	
	The special case of \srs where all machine speeds are either~$0$ or~$1$ has
	been studied previously by Stein and Zhong~\cite{SteinZ20}. They introduced the
	problem with identical machines and an unknown number of machines that fail (speed
	$0$) in the second stage. They present a simple lower bound of~$4/3$ on the
	robustness factor with equal jobs and  design a general $5/3$-robust algorithm.
	For infinitesimal jobs, they give an improved $1.2333$-robust algorithm complemented by a lower bound for each number of machines which tends
	to~$(1+\sqrt{2})/2\approx 1.207$ for large~$m$. Stein and Zhong also consider
	the objective of minimizing the maximum difference between the most loaded and
	the least loaded machine,
	motivated by questions on fair allocation.

\subsection*{Our Contribution}	

	We introduce the \srs problem and present robust algorithms. The 
	algorithmic difficulty of this problem is to construct bags in the first stage that are robust 
	under any choice of machine speeds in the second stage. The straightforward 
	approach of using any makespan-optimal solution on~$m$ identical machines is not 
	sufficient. \Cref{lem:optnotrobust} shows that such an algorithm might 
	have an arbitrarily large robustness factor. Using \emph{Longest Processing Time} 
	first (\lpt) to create bags does the trick and is $(2 - \frac1m)$-robust 
	for arbitrary job sizes (\Cref{thm:speedLPT}). While this was known for speeds in~$\{0,1\}$~\cite{SteinZ20}, our result for general speeds  
	is much less obvious.
	
	Note that \lpt aims at ``balancing'' the bag sizes which cannot lead to a
	better robustness factor than~$2-\frac1m$ as we show
	in~\Cref{lem:balancedlowerbound}. Hence, to improve upon this factor, we need to
	carefully construct bags with imbalanced bag sizes. There are two major
	challenges with this approach: (i) finding the ideal imbalance in the bag sizes
	independent from the actual job processing times that would be robust for all
	adversarial speed settings simultaneously and (ii) to adapt bag sizes to
	accommodate discrete jobs.
	
	A major contribution of this paper is an optimal solution to the first challenge by considering infinitesimal jobs (\cref{thm:speedsandLB,thm:speedsandUB}). One can think of this as filling bags with \emph{sand} to the desired level. In this case, the speed-robust scheduling problem boils down
	to identifying the best bag sizes as placing the jobs into bags becomes trivial.
	We give, for any number of machines, optimally imbalanced bag sizes and prove a
	best possible robustness factor of
	\[
		\rhosand = \frac{\nrbags^\nrbags}{\nrbags^\nrbags-(\nrbags-1)^\nrbags} \leq
		\frac{e}{e-1} \approx 1.58\,.
	\] 
	For infinitely many infinitesimal jobs in the particular machine environment in which all
	machines have either speed~$0$ or~$1$, we obtain an algorithm (\cref{theo:UBSandOnParallel}) with robustness
	factor
	\[
		\rhosandid(m) = \max_{t\leq \frac m2,~t\in\mathbb N}\ \frac{1}{\frac t{m-t}
		+\frac{m-2t}m} \ \leq\  \frac{1+\sqrt{2}}{2} \approx 1.207 = \rhoid\,.
	\]
	This improves the previous upper bound of~$1.233$ by Stein and
	Zhong~\cite{SteinZ20} and matches exactly their lower bound for each~$m$.
	Furthermore, we show that the lower bound in \cite{SteinZ20} holds even for
	randomized algorithms (\cref{thm:speedsandLB}),  
	so our algorithm is optimal for both,
	deterministic and randomized scheduling.

	The above tight results for infinitesimal jobs are crucial for our further
	results for discrete jobs. Following the figurative notion of sand for infinitesimal jobs, we think of equal-size jobs as \emph{bricks} and arbitrary jobs as \emph{rocks}. Building on those ideal bag sizes, our approaches
	differ substantially from the methods in \cite{SteinZ20}. When all jobs have equal processing time, we obtain a $1.8$-robust solution through a careful analysis of the trade-off
	between using slightly imbalanced bags and a scaled version of the ideal bag sizes computed for the infinitesimal setting (\Cref{theo:speedbrickUB}).
	
	When machines have only speeds in~$\{0,1\}$ and jobs
	have arbitrary equal sizes, i.e., unit size, 
	we give an optimal $\frac{4}{3}$-robust algorithm (\Cref{thm:BricksUB}). This is
	an interesting class of instances as the best known lower bound of~$\frac43$ for
	discrete jobs uses only unit-size jobs~\cite{SteinZ20}. To achieve this 
	result, we, again, crucially exploit the ideal bag sizes computed for
	infinitesimal jobs by using a scaled variant of these sizes. Some cases,
	depending on~$m$ and the optimal makespan on~$m$ machines, have to be handled
	individually. Here, we use a direct way of constructing bags with at most four
	different bag sizes, and some cases can be solved by an integer linear program.
	We summarize our results in Table~\ref{tab:results}.
		
	Inspired by traditional one-stage scheduling problems where jobs have {\em
	machine-dependent execution times} (unrelated machine scheduling), one might ask
	for such a generalization of our problem. However, it is easy to rule out any robustness 
	factor for such a setting: Consider four machines and five
	jobs, where each job may be executed on a unique pair of machines. Any algorithm
	must build at least one bag with at least two jobs. For this bag, there is
	at most one machine to which it can be assigned with finite makespan. If this
	machine fails, the algorithm cannot complete the jobs whereas an optimal
	solution can split this bag on multiple machines to get a finite makespan.

\begin{table}
	\renewcommand*\arraystretch{1.5}
	\centering	
	\def\arraystretch{1.1}
	\begin{tabular}{lcccc}	
		&\multicolumn{2}{c}{General speeds}& \multicolumn{2}{c}{Speeds from~$\{0,1\}$ } \\
		\cmidrule{2-5}
		&Lower bound&Upper bound&Lower bound&Upper bound\\
		\midrule
		Discrete jobs&~$\rhosandm m~$&~$2 - \tfrac1m$ &$ \tfrac43$&~$\tfrac53$  \\
		(Rocks)\footnotesize & (\cref{thm:speedsandLB}) & (\cref{thm:speedLPT}) & \cite{SteinZ20} & \cite{SteinZ20} \\
		\midrule
		Equal-size jobs&$\rhosandm m~$&$1.8$&\multicolumn{2}{c}{$\tfrac43$} \\
		(Bricks) & (\cref{thm:speedsandLB}) & (\cref{theo:speedbrickUB}) & \multicolumn{2}{c}{(\cite{SteinZ20}, \cref{thm:BricksUB})} \\ 
		\midrule
		Infinitesimal jobs& \multicolumn{2}{c}{$\rhosandm m \leq\tfrac{e}{e-1} \approx 1.58$} 
		& \multicolumn{2}{c}{$\rhoid(m) \leq \tfrac{1+\sqrt{2}}{2} \approx 1.207$} \\
		(Sand) & \multicolumn{2}{c}{(\cref{thm:speedsandLB,thm:speedsandUB})} 
		& \multicolumn{2}{c}{(\cite{SteinZ20}, \cref{theo:UBSandOnParallel})} \\ 
		\bottomrule
	\end{tabular}
	\caption{Summary of results on \srs.}\label{tab:results}
\end{table}

\section{Speed-Robust Scheduling with Infinitesimal Jobs}\label{sec:Sand}

	In this section, we assume that there are infinitely many jobs of infinitesimal processing time, we say {\em infinitesimal jobs}.
	We give optimal algorithms for \srs for both, the general case (\Cref{sec:SandOnRelated}) and the
	special case with speeds in~$\{0,1\}$ (\Cref{sec:SandOnParallel}).

\subsection{General Speeds}\label{sec:SandOnRelated}

	We present an algorithm for \srs with
		infinitesimal jobs that achieves a best-possible robustness factor of~$\rhosand$ for all~$m\geq 1$, where 
			\[
				\rhosand = \frac{\nrbags^\nrbags}{\nrbags^\nrbags-(\nrbags-1)^\nrbags} 
				\leq \frac{e}{e-1} \approx 1.58\,.
			\]

	We first show that, even when we
	restrict the adversary to a particular set of speed configurations, no deterministic
	algorithm can achieve a robustness factor better than~$\rhosand$. Note that
	since we can scale all speeds equally by an arbitrary factor without
	influencing the robustness factor, we can assume that the sum of the speeds is
	equal to~$1$. Similarly, we can assume that the total processing time of the
	jobs is equal to~$1$, such that the optimal makespan of the adversary is equal
	to~$1$ and the worst-case makespan of an algorithm is equal to its robustness
	factor.
	
	Intuitively, the idea
	behind the set of $m$ speed configurations is that the adversary can set
	$\nrbags-1$ machines to equal low speeds and one machine to a high speed. The low speeds are set such
	that one particular bag size just fits on that machine when aiming for the given robustness
	factor. This immediately implies that all larger bags have to be put on the fast
	machine together. This way, the speed configuration can \emph{target} a certain 
	bag size. We provide specific bag sizes that achieve a robustness 
	of~$\rhosand$ and show that for the speeds targeting these bag sizes, other bag sizes would result in even larger robustness factors.
	
	We define~$U=\nrbags^\nrbags$,~$L=\nrbags^\nrbags-(\nrbags-1)^{\nrbags}$, and
	$t_k = (\nrbags-1)^{\nrbags-k} \nrbags^{k-1}$ for~$k\in\{1,\ldots,\nrbags\}$.
	Intuitively, these values are chosen such that the bag sizes~$t_i/L$ are optimal 
	and~$t_i/U$ corresponds to the low speed of the~$i$-th speed configuration.
	It is easy to verify that~$\rhosand=U/L$ and for all~$k$ we have
	\begin{equation}\label{eq:ULt}
		\sum_{i<k} t_i =  (\nrbags-1) t_{k} -U + L\,. 	
	\end{equation} 
	In particular, this implies that~$\sum_{i\leq \nrbags} t_i = \nrbags t_\nrbags -U+L = L$
	and, hence, that the sum of the bag sizes is~$1$.
	Let~$a_1\le\cdots\le a_\nrbags$ denote the bag sizes chosen by an algorithm 
	and~$s_1\le\cdots\le s_\nrbags$ the speeds chosen by the adversary.
	\begin{theorem}
	\label{thm:speedsandLB}
		For any~$\nrbags\geq 1$, no deterministic algorithm for \srs
		with
		infinitesimal jobs can have a robustness factor less than~$\rhosand$. 
	\end{theorem}
	\begin{proof}
		We restrict the adversary to the following~$\nrbags$ speed configurations 
		indexed by~$k\in\{1,\ldots,\nrbags\}$:
		\[ 
			\mc S_k := \big\{s_1 = t_k/U,~ s_2 = t_k/U,~ \dots ,~ s_{\nrbags-1} 
			= t_k/U,~ s_\nrbags = 1 - (\nrbags-1) t_k/U\big\}\,.
		\]
		Note that for all~$k\in\{1,\ldots,\nrbags\}$, we have~$\nrbags t_k \leq U$ and, thus,~$s_\nrbags\geq s_{\nrbags-1}$.
		
		We show that for any bag sizes~$a_1,\ldots,a_\nrbags$, the adversary can force 
		the algorithm to have a makespan of at least~$U/L$ with some~$\mc S_k$. Since the 
		optimal makespan is fixed to be equal to~$1$ by assumption, this implies a robustness 
		factor of at least~$U/L$.
		
		Let~$k^\star$ be the smallest index such that~$a_k\ge t_k/L$. Such an index exists
		because the sum of the~$t_i$'s is equal to~$L$ (\Cref{eq:ULt}) and the sum of
		the~$a_i$'s is equal to~$1$. Now, consider the speed configuration~$\mc S_{k^\star}$. If
		one of the bags~$a_i$ for~$i\geq k^\star$ is not scheduled on the~$\nrbags$-th machine,
		the makespan is at least~$a_i/s_1\geq a_{k^\star}U/t_{k^\star} \geq U/L$. Otherwise, all~$a_i$
		for~$i\geq k^\star$ are scheduled on machine~$\nrbags$. Then, using Equation~\eqref{eq:ULt}, the load on that machine 
		is at least
		\begin{align*}
			\sum_{i\geq k^\star} a_i = 1 - \sum_{i< k^\star} a_i  \geq 1 - \frac 1L \sum_{i< k^\star} t_i   
			= \frac{1}{L}\left(L - (\nrbags-1)t_{k^\star} + U - L\right)
			= \frac UL s_\nrbags\,.
		\end{align*}
		Thus, either a machine~$i<m$ with a bag~$i'\geq k^*$
	    or machine~$i=\nrbags$ has a load of at least~$s_i \cdot U/L~$ and determines the makespan.
	\end{proof}
	
	For given bag sizes, we call a 
	speed configuration that maximizes the minimum makespan a \emph{worst-case speed configuration}.
	Before we provide the strategy that obtains a matching robustness factor, 
	we state a property of such best strategies for the adversary. 
	\begin{restatable}{lemma}{LemmaFullAssignment}\label{clm:fullassignment}
		 Given bag sizes and a worst-case speed configuration, for each machine~$i$, there exists an optimal assignment of the bags to the machines such that only machine~$i$ determines the makespan. 
	\end{restatable}
	\begin{proof}
		Consider a given set of bag sizes and a speed configuration~$\{s_1,\dots,s_m\}$ 
		that maximizes the minimum makespan for those bag sizes. Let~$\cmax^*$ be the minimum makespan of the best
		assignment of the bags given these speeds. This implies that, for any other speed configuration,
		there exists an assignment which has a makespan at most~$\cmax^*$. 
		
		We prove the lemma by contradiction. If there exists a machine~$i$ that does
		not satisfy the lemma, we increase its speed by an additive factor
		of~$\varepsilon$ and we decrease the speed of all other machines
		by~$\varepsilon/(m-1)$. Pick~$\varepsilon$ such that all (non-optimal)
		assignments that cause the load of machine~$i$ to be strictly greater
		than~$\cmax^* \cdot s_i$ still satisfy that their respective load on
		machine~$i$ is strictly greater than ~$\cmax^* \cdot (s_i + \varepsilon)$.
		Denote the new speeds by~$s_{i'}'$ for~$1 \leq i'\leq m$. Now, consider any
		assignment. If the load of machine~$i$ is larger than~$\cmax^*\cdot s_i$, it
		is also larger than~$\cmax^*\cdot s_i'$ by the construction of~$s_i'$ and the assignment still has a makespan strictly larger than~$\cmax^*$.
		Conversely, if the load of machine~$i$ is at most~$\cmax^*\cdot s_i$,
		there must be a machine~$i'\neq i$ with load at
		least~$\cmax^*\cdot s_{i'}$. Otherwise, either~$i$ is the only machine that admits the makespan and already satisfies the lemma, or the assignment has a makespan smaller than~$\cmax^*$, a contradiction.
		Consider any assignment where at least one
		other machine~$i'$ has load at least~$\cmax^*\cdot s_{i'}$. Since we decreased
		the speed of all machines except machine~$i$, the load of machine~$i'$ is
		strictly larger than~$\cmax^*\cdot s_{i'}'$ leading to a makespan strictly
		greater than~$\cmax^*$. This contradicts that the speed configuration
		maximizes the minimum makespan since every assignment with the new speeds has
		a makespan strictly larger than~$\cmax^*$.
	\end{proof}

	Note that, by~\Cref{clm:fullassignment}, for a worst-case speed configuration, many 
	bag-to-machine assignments obtain the 
	optimal makespan. \Cref{clm:fullassignment} also implies that, for such a
	speed configuration, all speeds are non-zero.
	Indeed, if a machine has a speed equal to zero, then it cannot determine the makespan in an optimal assignment (a better speed configuration would slow down other machines to increase its speed).
	
	Let \sandalg denote the algorithm that creates~$\nrbags$ bags of the following sizes
	\[
		a_1 = t_1/L,~ a_2 = t_2/L,~ \dots ,~ a_m = t_m/L\,.
	\]
	Note that this is a valid algorithm since the sum of these bag sizes is equal to~$1$.
	Moreover, these bag sizes are exactly such that if we take the speed configurations 
	from the proof of \Cref{thm:speedsandLB}, placing bag~$j$ on a slow machine in configuration~$j$
	results in a makespan that is equal to~$\rhosand$.
	
	We proceed to show that \sandalg has a robustness factor of~$\rhosand$.
	\begin{theorem}\label{thm:speedsandUB}
		For any~$m\geq 1$, \sandalg is $\rhosand$-robust for \srs with
		 infinitesimal jobs.
	\end{theorem}
	\begin{proof}
		Let~$a_1,\ldots,a_\nrbags$ be the bag sizes as specified by \sandalg and let
		$s_1,\ldots,s_\nrbags$ be a speed configuration that maximizes the minimum
		makespan given these bag sizes. Further, consider an optimal assignment of
		bags to machines and let~$\cmax^*$ denote its makespan. We use one particular
		(optimal) assignment to obtain an upper bound on~$\cmax^*$. By
		\Cref{clm:fullassignment}, there exists an optimal assignment where only
		machine~$1$ determines the makespan, i.e., machine~$1$ has load~$\cmax^*\cdot
		s_1$ and any other machine~$i$ has load strictly less than~$\cmax^* \cdot
		s_i$. Consider such an assignment. If there are two bags assigned to
		machine~$1$, then there is an empty machine with speed at least~$s_1$.
		Therefore, we can put one of the two bags on that machine and decrease the
		makespan. This contradicts~$\cmax^*$ being the optimal makespan, so there is
		exactly one bag assigned to machine~$1$. Let~$k$ be the index of the unique
		bag placed on machine~$1$, i.e.,~$\cmax^*=a_k/s_1$, and let~$\ell$ be the
		number of machines of speed~$s_1$.
				
		If~$a_k>a_\ell$, machine~$i\in\{1,\ldots,\ell\}$ with speed~$s_1$ can
		be assigned bag~$i$ with a load that is strictly less than~$\cmax^*\cdot s_1$.
		Thus, given the current assignment, we can remove bag~$a_k$ from machine~$1$ and place
		the~$\ell$ smallest bags on the~$\ell$ slowest machines, one per machine, e.g., bag~$a_i$ on machine~$i$ for~$i\in\{1,\ldots,\ell\}$. This empties at least one
		machine of speed strictly larger than~$s_1$. Then, we can place bag~$a_k$ on this, now
		empty, machine, which yields a makespan that is strictly smaller than~$\cmax^*$. This
		contradicts the assumption that~$\cmax^*$ is the optimal makespan and,
		thus,~$a_k\le a_\ell$, which implies~$k \leq \ell$.
		
		Let~$\machineload_i$ denote the total processing time of bags that are
		assigned to machine~$i$ and let~$C$ be the total \emph{remaining capacity} of
		the assignment, that is,~$C:=\sum_{i=1}^{m} (s_i\cmax^*-\machineload_i)$. We
		construct an upper bound on~$C$, which allows us to bound~$\cmax^*$.
		
		Machines in the set~$\{2,\ldots,\ell\}$ cannot be assigned a bag of size
		larger than~$a_k$ since their load would be greater than~$\cmax^* \cdot s_1$,
		causing a makespan greater than~$\cmax^*$. Therefore, we assume without loss
		of generality that all bags~$a_j<a_k$ are assigned to a machine with
		speed~$s_1$. The total remaining capacity on the first~$k$ machines is
		therefore equal to~$(k-1) a_k - \sum_{i< k} a_i$.
		
		Consider a machine~$i>k$. If its remaining capacity  
		is greater than~$a_k$, then 
		we can decrease the makespan of the assignment by moving bag~$k$ to 
		machine~$i$. Therefore, the remaining capacity on machine~$i$ is at most~$a_k$.
			
		Combining the above and using \eqref{eq:ULt}, we obtain:
		\begin{align*}
			C &\leq (m-1) a_k - \sum_{i< k} a_i = \frac{1}{L}\left( (m-1) t_k - \sum_{i<
			k} t_i\right) = \frac{1}{L}\left(U-L\right)\,.
		\end{align*}
		The total processing time is~$\sum_{i=1}^m a_i = 1$, and the maximum total 
		processing time the machines could process with makespan~$\cmax^*$
		is~$\sum_{i=1}^m s_i \cmax^* = \cmax^*$. Since the latter is equal to the total processing time 
		plus the remaining capacity, we have~$\cmax^* = 1+C \leq U/L$, which proves the~lemma.
	\end{proof}
	
	While the robustness factor~$\rhosandm{m}$ is best possible for every~$m$ for any deterministic algorithm, this is not true
	when we allow
	algorithms that make randomized decisions and compare to an oblivious
	adversary. For~$m=2$, uniformly randomizing between bag sizes~$a_1=a_2=1/2$ 
	and~$a_1=1/4$,~$a_2=3/4$ yields a robustness factor that is slightly better 
	than~$\rhosandm{2} = 4/3$. Interestingly, with speeds in~$\{0,1\}$, the optimal
	robustness factor is equal for deterministic and randomized~algorithms as we show in \Cref{thm:BricksLBrand}.

\subsection{Speeds in \texorpdfstring{$\{0,1\}$}{0,1}}\label{sec:SandOnParallel}
	This section is devoted to showing that 	the best-possible robustness factor that can be achieved with speeds in~$\{0,1\}$ 
	is precisely
	\[
		\rhosandid(m) = \max_{t\in\N,~t\leq \frac m2}\ \frac{1}{\frac t{m-t} 
					+\frac{m-2t}m} \ \leq\  \frac{1+\sqrt{2}}{2} = \rhoid \approx 1.207\, 
	\]
	for both, deterministic and randomized algorithms. The main contribution is the upper bound, that is, the following theorem.
	
	\begin{restatable}{theorem}{IdenticalWSand}\label{thm:IdenticalWSand}\label{theo:UBSandOnParallel}
		For all~$m\geq 1$, there is a deterministic $\rhoid(m)$-robust algorithm for \srs with speeds in~$\{0,1\}$ for infinitesimal jobs.
\end{restatable}

	The
	lower bound for deterministic algorithms and some useful insights were
	already presented in~\cite{SteinZ20}. We recall some of these
	insights here because they are used in the proof. To do so, we introduce 
	some necessary notation used in the remainder of this paper. The number
	of failing machines (i.e., machines with speed equal to~$0$) is referred to as~$t\geq0$, and
	we assume w.l.o.g.\ that these are machines~$1,\dots,t$. Furthermore, we assume
	for this subsection again w.l.o.g.\ that the total volume of infinitesimal jobs
	is~$m$, and we will define bags~$1,\dots,m$ with respective sizes
	$a_1\leq\dots\leq a_m$ summing to at least~$m$ (the potential excess being unused).	
	\begin{lemma}[Statement (3) in~\cite{SteinZ20}]\label{lem:folding}
	 	For all~$m\ge1$ and~$t \leq m/2$, there exists a makespan-minimizing allocation of bags to
	 	machines for \srs with speeds in~$\{0,1\}$ and infinitely many infinitesimal jobs that assigns the smallest~$2t$ bags to machines~$t+1,\ldots,2t$.
	\end{lemma}

	Since \Cref{lem:folding} only works for~$t\leq m/2$, one may worry that, for larger~$t$,
	there is a more difficult structure to understand. The following insight shows
	that this worry is unjustified. Indeed, if~$\nrmach< m/2$ is the number of
	machines that do \emph{not} fail, one can simply take the solution for
	$2\nrmach$ machines and assign the bags from any two machines to one machine.
	The optimal makespan is doubled and that of the algorithm is at most
	doubled, conserving the robustness.
	
	\begin{lemma}[Proof of Theorem 2.2 in~\cite{SteinZ20}]\label{lem:m-half}
		Let~$\rho>1$. For all~$m\ge1$, if an algorithm is $\rho$-robust for 
		\srs with speeds in~$\{0,1\}$ and infinitely many infinitesimal jobs 
		for~$t\leq m/2$, it is $\rho$-robust for $t\leq m-1$.
	\end{lemma}
	
	We will thus focus on computing bag sizes such that the makespan of a best
	allocation according to Lemma~\ref{lem:folding} is within a~$\rhosandid(m)$
	factor of the optimal makespan when~$t\leq m/2$. The approach
	in~\cite{SteinZ20} to obtain the (as we show tight) lower bound~$\rhosandid(m)$
	is as follows. Given some~$t\leq m/2$ and a set of bags allocated according to Lemma~\ref{lem:folding},
	\begin{enumerate}
		\item[(i)] the load 
		on machines~$t+1,\dots,2t$ is at most~$\rhosandid(m)$
		times the optimal makespan~\mbox{$m/(m-t)$}, and
		\item[(ii)] the load 
		on machines~$2t+1,\dots,m$ is a most
		$\rhosandid(m)$ because those machines only hold a \emph{single} bag after a
		simple ``folding'' strategy for assigning bags to machines, which we define
		below.
	\end{enumerate}
	In particular, since~$t=0$ is possible, (ii) implies that all bag sizes are at most
	$\rhosandid(m)$.
	The fact that the total processing volume of~$m$ has to be accommodated and maximizing over
	$t$ results in the lower bound given in Theorem~\ref{thm:IdenticalWSand}.
	
	To define the bag sizes leading to a matching upper bound, we further
	restrict our choices when~$t \leq m/2$ machines fail. Of course, since we match
	the lower bound, the restriction is no limitation but rather a simplification.
	When $t \leq m/2$ machines fail, we additionally assume that the machines~$t+1,
	\ldots, 2t$ receive exactly two bags each: Assuming~$t\leq m/2$, the
	\simplefold of these bags onto machines assigns bags~$i\geq t+1$ to machine
	$i$, and bags~$i=1,\dots,t$ (recall machine~$i$ fails) to machine~$2t-i+1$.
	Hence, bags~$1,\dots,t$ are ``folded'' onto machines~$2t,\dots,t+1$ (sic),
	visualized in Figure~\ref{fig:folding}. 
	
	For given~$m$, let~$t^\star$ be an optimal adversarial choice for~$t$ in
	\cref{theo:UBSandOnParallel}. Assuming there are bag sizes~$a_1,\ldots,a_m$
	that match the bound~$\rhosandid(m)$ through simple folding, by (i) and (ii), we
	precisely know the makespan on all machines after folding when~$t=t^\star$.
	That fixes~$a_{i}+a_{2t+1-i}=\rhosandid(m)\cdot m/(m-t)$ for all~$i=1,\dots,t$
	and~$a_{2t+1},\dots,a_m=\rhosandid(m)$, see Figure~\ref{fig:folding}. In
	contrast to~\cite{SteinZ20}, we show that defining~$a_i$ for~$i=1,\dots,t$ to
	be essentially a linear function of~$i$, and thereby fixing all bag sizes,
	suffices to match~$\rhosandid(m)$. The word ``essentially'' can be dropped when
	replacing~$\rhosandid(m)$ by~$\rhosandid$.

	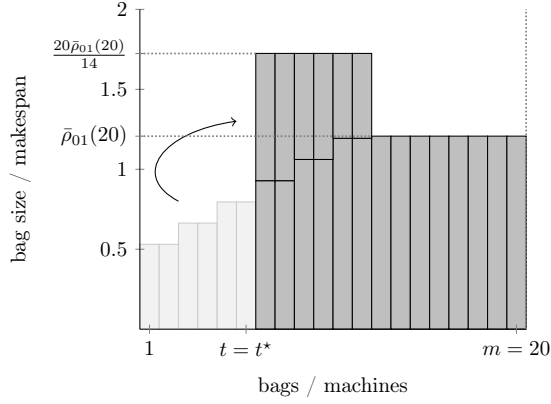
\begin{figure}
		\centering
		\begin{tikzpicture}[fct/.style={very thick, domain=0:1}, samples=100, font=\small, scale=0.8]
			\begin{scope}[xscale=6.42,yscale=2.66]
				\draw[fill=black!5,draw=black!20] (0,0) rectangle (0.05,0.530504);
				\draw[fill=black!5,draw=black!20] (0.05,0) rectangle (0.1,0.530504);	
				\draw[fill=black!5,draw=black!20] (0.1,0) rectangle (0.15,0.66313);
				\draw[fill=black!5,draw=black!20] (0.15,0) rectangle (0.2,0.66313);	
				\draw[fill=black!5,draw=black!20] (0.2,0) rectangle (0.25,0.795756);
				\draw[fill=black!5,draw=black!20] (0.25,0) rectangle (0.3,0.795756);	
				\draw[fill=black!25] (0.3,0) rectangle (0.35,0.928382);
				\draw[fill=black!25] (0.35,0) rectangle (0.4,0.928382);	
				\draw[fill=black!25] (0.4,0) rectangle (0.45,1.06101);
				\draw[fill=black!25] (0.45,0) rectangle (0.5,1.06101);	
				\draw[fill=black!25] (0.5,0) rectangle (0.55,1.19363);
				\draw[fill=black!25] (0.55,0) rectangle (0.6,1.19363);	
				\draw[fill=black!25] (0.6,0) rectangle (0.65,1.2069);
				\draw[fill=black!25] (0.65,0) rectangle (0.7,1.2069);	
				\draw[fill=black!25] (0.7,0) rectangle (0.75,1.2069);
				\draw[fill=black!25] (0.75,0) rectangle (0.8,1.2069);
				\draw[fill=black!25] (0.8,0) rectangle (0.85,1.2069);
				\draw[fill=black!25] (0.85,0) rectangle (0.9,1.2069);						
				\draw[fill=black!25] (0.9,0) rectangle (0.95,1.2069);
				\draw[fill=black!25] (0.95,0) rectangle (1,1.2069);
				
				\draw[fill=black!25] (0.3,0.928382) rectangle (0.35,1.7241);
				\draw[fill=black!25] (0.35,0.928382) rectangle (0.4,1.7241);
				\draw[fill=black!25] (0.4,1.06101) rectangle (0.45,1.7241);
				\draw[fill=black!25] (0.45,1.06101) rectangle (0.5,1.7241);
				\draw[fill=black!25] (0.5,1.19363) rectangle (0.55,1.7241);
				\draw[fill=black!25] (0.55,1.19363) rectangle (0.6,1.7241);
				
				\draw[->] (0.1,0.8) to [bend left = 55] (0.25,1.3);
				%
			\end{scope}					
			\begin{axis}[	
				axis lines=middle,
				axis line style={-},
				ylabel near ticks,
				xlabel near ticks,
				ylabel = {bag size / makespan},
				xlabel = {bags / machines},
				xtick = \empty,
				extra x ticks = {0.025,0.275,0.975},
				extra x tick labels = {1,$t=t^\star$,$m=20$},			
				extra y ticks={1.2069,1.7241},
				extra y tick labels={$\bar{\rho}_{01}(20)
					$,$\frac{20\bar{\rho}_{01}(20)}{14} 
					$},
				domain=0:1, ymin=0, xmax=1, xmin=0, ymax=2,
				width=8cm]
				\addplot[domain=0:2.36, gray, densely dotted, thick] (1,x);
				\addplot[domain=0:0.6, gray, densely dotted, thick] (x,1.207);	
				\addplot[domain=0:0.3, gray, densely dotted, thick] (x,1.724);
			\end{axis}
		\end{tikzpicture}
		\caption{Folding optimally sized bags when~$t^\star$ machines fail.}\label{fig:folding}
	\end{figure}
		
	A clean way of thinking about the bag sizes is through \emph{profile functions}
	which reflect the distribution of load over bags in the limit case
	$m\rightarrow\infty$. Specifically, we identify the set~$\{1,\dots,m\}$ with
	the interval~$[0,1]$ and define a continuous non-decreasing profile function
	$\bar{f}:[0,1]\rightarrow\mathbb{R_+}$ integrating to~$1$. A simple way of
	getting back from the profile function to actual bag sizes of total size
	approximately~$m$ is by equidistantly ``sampling''~$\bar{f}$, that is, by defining
	$a_i:=\bar{f}(\frac{i-1/2}m)$ for all~$i$.
	
	Our profile function~$\bar f$ implements the above observations and ideas in
	the continuous setting. Indeed, our choice
	\[
		\bar f(x)= \min\left\{\frac12+\rhoid\cdot x,\rhoid\right\}=\min\left\{\frac12+\frac{(1+\sqrt2)\cdot x}{2},\frac{1+\sqrt2}{2}\right\}
	\]
	is linear up to~$\beta = 2-\sqrt{2}$, which turns out to be equal to $\lim_{m\rightarrow\infty}2t^\star/m$, and then constantly equal to $(1+\sqrt{2)}/2$ since $$\frac12+\frac{(1+\sqrt2)\cdot x}{2}\leq \frac{1+\sqrt2}{2}\;\Leftrightarrow\; x\leq \frac{\sqrt{2}}{1+\sqrt{2}} = 2-\sqrt{2}.$$ We give some intuition
	for why this function yields the desired bound using the continuous counterpart of folding. When
	$t\leq t^\star$ machines fail, i.e., a continuum of machines with measure
	$x\leq \beta/2$, we fold the corresponding part of~$\bar f$ onto the interval
	$[x,2x]$, yielding a rectangle of width~$x$ and height~$\bar f(0)+\bar
	f(2x)=2\bar f(x)$. We have to prove that the height does not exceed the optimal
	makespan~$1/(1-x)$ by more than a factor of~$\rhoid$. Equivalently, we maximize~$2 \bar f(x)(1-x)$ (even over~$x \in \R$) and observe the maximum of~$\rhoid=(1+\sqrt2)/2$ at
	$x=\beta/2$. When~$x\in(\beta/2,1/2]$, note that by folding, we \emph{still}
	obtain a rectangle of height~$2\bar f(x)$ (but width~$\beta-x$), dominating the
	load on the other machines. Hence, the makespan is at most~$\rhoid/(1-x)$ for every~$x \in [0,1/2]$.
		
	\begin{figure}
		\centering
		\begin{tikzpicture}	[fct/.style={thick}, font=\small, scale = 0.8]			
			\begin{axis}[
				axis lines=middle,
				axis line style={-},
				ylabel near ticks,
				xlabel near ticks,
				ylabel = {$\bar f$},
				xlabel = {$x$},
				xtick = \empty,
				extra y ticks = {1.2071},
				extra y tick labels = {$\bar{\rho}_{01}
					$},			
				extra x ticks = {0.2,0.4,0.586,0.8,1},
				extra x tick labels = {$0.2$,$0.4$,$\beta\approx0.586$,$0.8$,$1$},			
				domain=0:1, ymin=0, xmax=1, xmin=0, ymax=1.5,
				width=8cm]
				\addplot[ybar, bar width=0.05, samples=12, domain=0.025:0.575,fill=black!25] (x,0.5+1.2071*x);
				\addplot[ybar, bar width=0.05, samples=8, domain=0.625:0.975,fill=black!25] (x,1.2071);
				\addplot[fct, domain=0:0.586] (x,0.5+1.2071*x) node[pos=0.6,above] {$\bar f$};
				\addplot[fct, domain=0.586:1] (x,1.2071);
				\addplot[domain=0:2.36, densely dotted, ultra thick] (1,x);	
				\addplot[domain=0:1.2071, densely dotted, very thick] (0.586,x);				
				\addplot[domain=0:0.586, densely dotted, very thick] (x,1.2071);	
			\end{axis}	
		\end{tikzpicture}	
		\caption{The profile function~$\bar f$ and the equidistant ``sampling'' 
			of $\sandalg_{01}$.} \label{fig:sampling}
	\end{figure}
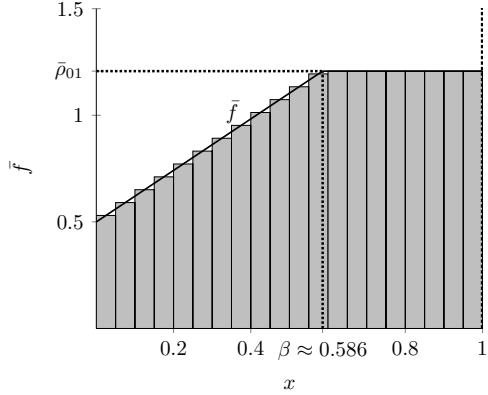
	
	Directly ``sampling''~$\bar f$, we obtain a bound of $\rhoid$, stated below.  
	Later, we make use of the corresponding simpler algorithm.  
	Let $\sandalg_{01}$ denote the algorithm that creates~$m$ bags of size
	$a_i:=\bar{f}\big(\frac{i-1/2}m\big)$, for~$i \in \{1,\ldots,m\}$. 
	See Figure~\ref{fig:sampling} for a visualization.
	
\begin{theorem}\label{prop:identical-steps} 
	$\sandalg_{01}$ is $\rhoid$-robust for \srs with speeds in $\{0,1\}$ and 
	infinitely many infinitesimal jobs for all~$m\geq 1$.
\end{theorem}

\begin{proof}
	Our proof naturally splits into two parts. In the first part, we show that the bag sizes are
	feasible, i.e., their total size is at least $m$. In the second part, we show that the bag sizes 
	achieve the claimed robustness factor.	
	
	To show that the bag sizes are feasible, we prove that 
	\begin{equation}\label{eq:sampling}
	a_i\geq m\cdot\int_{\frac{i-1}m}^{\frac im} \bar f(x) \dx
	\end{equation}	
	for all $i\in\{1,\dots,m\}$. Then the first part follows by summing
	over~\eqref{eq:sampling} for all~$i$ and indeed, as required for a profile
	function, $\bar f$ integrates to $1$:
	\begin{align*}
	\int_0^1 \bar f(x) \dx & = \int_0^\beta \left(\frac12 + \rhoid x \right)\dx + \int_{\beta}^{1} \rhoid \dx \\
	&  = \frac{\beta}{2} +\frac{\rhoid \beta^2}{2} + (1-\beta)\cdot\rhoid  = 1 \, .  
	\end{align*}
	For $i$ with $\beta\notin(\frac{i-1}m,\frac im)$, we have that $\bar f$ is linear on $[\frac{i-1}m,\frac im]$ and therefore
	\[
	\int_{\frac{i-1}{\nrbags}}^{\frac{i}{\nrbags}} \bar f(x) \dx = \left(\frac im
	- \frac{i-1}{\nrbags}\right)\cdot \bar f\left(i-\frac
	12\right)=\frac{f(i-1/2)}m=\frac{a_i}m\, .
	\]
	For the single $i$ with $\beta\in(\frac{i-1}m,\frac im)$ --- and there is at
	least one such $i$ because $\beta$ is irrational --- we even get a stronger
	bound due to the fact that $\bar f$ is strictly concave on that interval.
	Formally, we distinguish two cases, in which we use either of the two linear
	functions from the definition of $\bar f$ as upper bound on~$\bar f$. 
	
	If
	$\frac{i-1/2}{\nrbags} \leq \beta$, it follows that
	\begin{align*}
	&\int_{\frac{i-1}{\nrbags}}^{\frac{i}{\nrbags}} \bar f(x) \dx \leq
	\int_{\frac{i-1}{\nrbags}}^{\frac{i}{\nrbags}} \left(\frac12 + \rhoid x\right)
	\dx
	= \left(\frac im - \frac{i-1}{\nrbags}\right)\cdot \bar f\left(i-\frac
	12\right)
	=\frac{a_i}m\, .		
	\end{align*}
	
	For $\frac{i-1/2}{\nrbags} > \beta$, we have that 
	\[
	\int_{\frac{i-1}{\nrbags}}^{\frac{i}{\nrbags}} \bar f(x) \dx \leq
	\int_{\frac{i-1}{\nrbags}}^{\frac{i}{\nrbags}} \rhoid  \dx = \left(\frac im -
	\frac{i-1}{\nrbags}\right)\cdot \bar f\left(i-\frac
	12\right)=\frac{f(i-1/2)}m=\frac{a_i}m\, .
	\]
	That finishes the proof of~\eqref{eq:sampling} for all $i$ and 
	verifies that our bag sizes are feasible.
	
	It remains to show that our bag sizes achieve the claimed robustness factor of
	$\rhoid\approx 1.207$. Essentially, the argument is a formal version of the
	intuitive argument we gave in the continuous setting, restricted to~$x$~(the
	measure of the continuum of failing machines) being $\frac{i-1/2}m$ for some
	$i\in\{1,\dots,m\}$. By Lemma~\ref{lem:m-half}, it suffices to consider the
	case that $t\leq m/2$ machines fail. By our self-imposed restriction, we only
	consider bag-to-machine assignments obtained through simple folding. Note that
	we only need to bound the load on machines that have two bags assigned to them
	after folding. Recall that these machines are machines $t+1,\dots,2t$; for all
	$i\in\{1,\dots,t\}$, bags $i$ and $2t+1-i$ are assigned to machine~$2t+1-i$.
	Also recall that
	\[
	a_i\leq\frac12+\rhoid\cdot\frac{i-\frac12}m
	\] 
	for all $i\in\{1,\dots,m\}$. Hence, the load created by bags $i\in\{1,\dots,t\}$ and $2t+1-i$ on machine~$2t+1-i$ is 
	\[
	a_i+a_{2t+1-i}\leq\frac12+\rhoid\cdot\frac{i-\frac12}m+\frac12+\rhoid\cdot\frac{2t+1-i-\frac12}m=1+2\rhoid\cdot \frac tm\,.
	\] 
	We would like to show that this load is at most a $\rhoid$ factor away from the load of the optimum, that is, 
	\[
	1+2\rhoid\cdot \frac tm\leq \rhoid\cdot\frac{m}{m-t}\,.
	\]
	Letting $x:=t/m$ yields the inequality 
	\[
	(1+2\rhoid x)(1-x)\leq \rhoid\,.
	\] 
	The left hand side of this inequality takes its maximum $\rhoid$ (even among all $x\in\mathbb{R}$) at~$x=\beta=2-\sqrt{2}$, which shows that the inequality is indeed true and therefore completes the proof.
\end{proof}

To show Theorem~\ref{thm:IdenticalWSand}, however, we have to match $\rhoid(m)$ for every $m$. To do so, we need to design bag sizes more carefully. Indeed, $\rhoid(m)$ is smaller than $\rhoid$ for all values of $m$, and therefore bags of sizes $\rhoid$ are not allowed anymore. For every $x\in[0,1]$, the bag size of the $\lceil x\cdot m\rceil$-th smallest bag still approaches~$\bar{f}(x)$ as $m\rightarrow\infty$, rather than being defined \emph{through} $\bar f$. Specifically, in the appendix we give a family of bag sizes parameterized by some $\delta>0$ that allow a simple computation of the robustness factor. The remainder of the proof is then concerned with showing algebraically that, for each $m$, $\delta$ can be chosen so as to fulfill the constraints imposed by feasibility and robustness.

	We close this section with showing that no better robustness factor can be achieved even by a randomized algorithm.
	
\begin{theorem}\label{thm:BricksLBrand}
		For all~$m\geq 1$ and $\varepsilon>0$, there is no randomized algorithm that achieves a robustness factor of $\rhoid(m)-\varepsilon$ against an oblivious adversary.
\end{theorem}

\begin{proof}
	The outline of this proof is based on the same result from~\cite{SteinZ20} for
	deterministic algorithms. Consider any randomized algorithm, the size of each
	bag follows some probability distribution which can be correlated. The problem
	can be described as follows: the adversary first selects the number $t$ of
	machine failures, knowing the distribution of the bag sizes; 
	then, the actual bag sizes are revealed; finally, the
	algorithm schedules these bags on $m-t$ machines. We assume by contradiction
	that, for every $t$, the expectation of the resulting makespan is smaller than
	$\rhosandid(m)/(1-\frac tm)$. 
	
	We consider an adversary with two possible strategies: make zero machines fail, or make
	$t\leq m/2$ machines fail, the value of $t$ being fixed later. For large~$m$,
	$t/m$ will approach $1-\sqrt 2/2$ and $\rhosandid(m)$ will approach \rhosandid.
	
	The expected size of each bag must be smaller than $\rhosandid(m)$; otherwise
	the expected makespan on $m$ machines would be too large. For every
	realization of bag sizes, there exists an optimal bag-to-machine allocation on
	$m-t$ machines that uses all machines, so has at least $m-2t$ machines
	containing a single bag. Reorder the machines so that machines $t+1$ to $m-t$
	have a single bag. The expected load of each of the first $t$ machines is
	smaller than $\rhosandid(m) / (1-\frac tm)$ as the optimal makespan on $m-t$
	machines is $1/(1-\frac tm)$. The expected load of each of the other machines
	is smaller than $\rhosandid(m)$ as they contain a single bag. By linearity of
	expectation, and due to the expected total load being equal to~$m$, we obtain
	the following contradiction:
	\begin{align*}
	m    &< \min_{t\leq \frac m2,~t\in\mathbb N}~~\frac{t\cdot \rhosandid(m)}{1-\frac tm} + (m-2t)\rhosandid(m)\\
	\rhosandid(m) & > \max_{t\leq \frac m2,~t\in\mathbb N} ~~\frac{1}{\frac t{m-t} +\frac{m-2t}m}~~ =~ \rhosandid(m)\,.
	\end{align*}
	This completes the proof.
\end{proof}

\section{Speed-Robust Scheduling with Discrete Jobs}\label{sec:JobsOnRelated}

	In this section, we consider the most general version of \srs, i.e., discrete jobs scheduled on machines with arbitrary unknown speeds. While in \Cref{sec:Sand,sec:Bricks} we crucially use in our algorithm design the assumption that jobs are infinitesimally small (sand) or of the same size (bricks), respectively, here, their sizes can vary arbitrarily (rocks).
	By a scaling argument, we
	may assume w.l.o.g.\ that the machine speeds satisfy~$\sum_{i=1}^m s_i =
	\sum_{j=1}^n p_j$.
	We first note in the following lemma that obtaining a robust algorithm is not
	trivial in this case, as even algorithms minimizing the largest bag size
	may not have a constant robustness factor. This contrasts with 
	the case where machine speeds are restricted to $\{0,1\}$, in which such algorithms are $(2-\frac2m)$-robust.
	To see this, once the number~$\nrmach$ of speed-1 machines is revealed, simply combine the
	two smallest bags repetitively if~$\nrmach<m$. The makespan is then at
	most twice the average load on~$\nrmach+1$ machines, so~$\frac{2\nrmach}{\nrmach+1}$
	times the average load on~$\nrmach$ machines.
	This is largest for $m'=m-1$ which gives the desired robustness.

	\begin{lemma}\label{lem:optnotrobust}
		There exists an algorithm for \srs minimizing the size of the largest bag which does not have a constant robustness factor.
	\end{lemma}
	\begin{proof}
		Consider any integer~$k\geq 1$, a number of machines~$m=k^2+1$,~$k^2$ unit-size
		jobs and one job of processing time~$k$. The maximum bag size is equal to~$k$,
		so an algorithm building~$k+1$ bags of size~$k$ respects the conditions of the
		lemma. Consider the speed configuration where~$k^2$ machines have speed~$1$ and
		one machine has speed~$k$. It is possible to schedule all jobs within a
		makespan~$1$ on these machines. However, the algorithm must either place a bag on
		a machine of speed 1 or all bags on the machine of speed~$k$, hence leading to
		a makespan of~$k$, and proving the result. Note that by adding~$k^2$ unit-size
		jobs, we can build a similar example where the algorithm does not leave
		bags empty, which is always beneficial.
	\end{proof}

	A feature, that is
	exploited in the lower bound, of the algorithms considered in \Cref{lem:optnotrobust} is that bags sizes are too unbalanced. A way
	to prevent this 
	would be to maximize the size of a minimum bag as
	well. But this criterion 
	becomes useless if we consider the same example as
	above with~$m=k^2+2$. Then, the minimum bag size is~$0$ as there are more
	machines than jobs, and the same lower bound holds.
	
	Hence, in order to obtain a robust algorithm in the general case, we focus on
	algorithms that aim at balanced bag sizes, for which the best lower bound is
	described in the following lemma.  An algorithm is called \emph{balanced} if, for an instance of unit-size jobs, the bag sizes created by the algorithm differ by at most one unit. In particular, a balanced algorithm creates~$m$ bags of size~$k$ when
	confronted with~$mk$ unit-size jobs and~$m$ bags. For balanced algorithms,
	we give a lower bound in Lemma~\ref{lem:balancedlowerbound} and a matching upper bound in Theorem~\ref{thm:speedLPT}.
	\begin{lemma}\label{lem:balancedlowerbound}
		No balanced algorithm for \srs can obtain a better robustness factor 
		than~$2-\frac{1}{m}$ for any~$m\geq 1$.
	\end{lemma}
	\begin{proof}
		Consider any~$m\geq 1$ and~$km$ unit-size jobs, with~$k=2m-1$. Assume the adversary
		puts~$m$ jobs on the first machine of speed $m$ and~$2m$ jobs on each of the remaining
		machines of speed $2m$ each. An algorithm that uses evenly balanced bags builds~$m$ bags of size
		$k$. It must either place a bag of size~$k$ on the machine of speed~$m$ or~$2k$
		jobs on a machine of speed~$2m$. In any case, the robustness factor is at
		least~$2-\frac1m$.
	\end{proof}
	
	We now show that this lower bound is attained by a simple algorithm, commonly
	named as {\em Longest Processing Time First} (\lpt) which considers jobs in
	non-increasing order of processing times and assigns each job to the bag that
	currently has the smallest size, i.e., the minimum allocated processing~time.
	\begin{theorem}\label{thm:speedLPT}
		\lpt is $(2-\frac{1}{m})$-robust for \srs for all~$m\ge1$.
	\end{theorem}
	\begin{proof}
		While we may assume that the bags are allocated optimally to the machines once
		the speeds are given, we use a different allocation for the analysis. This
	    can only worsen the robustness factor.
		
		Consider the~$m$ bags and let~$b$ denote the size of a largest bag,~$B$, that
		consists of at least two jobs. Place all bags of size strictly larger
		than~$b$, each containing only a single job, on the same machine
		as \opt{} places the corresponding  jobs. We define for each machine~$i$ with
		given speed~$s_i$ a capacity bound of~$(2-\frac{1}{m}) \cdot s_i$. Then, we
		consider the remaining bags in non-increasing order of bag sizes and iteratively assign
		them to the -- at the time of assignment -- 
		least loaded machine with sufficient remaining capacity.
	
		With the assumption~$\sum_{i=1}^m s_i = \sum_{j=1}^n p_j$ and the capacities~$(2-\frac{1}{m}) \cdot s_i$, it is sufficient to show that \lpt can
		successfully place all bags.
		
		The bags of size larger than~$b$ fit by definition as they contain a single job.
		Suppose for the sake of contradiction that there is a bag which cannot be assigned. Consider
		the first such bag and let~$T$ be its size. Let~$k<m$ be the number of bags
		that have been assigned already. Further, denote by~$w$ the size of a smallest
		bag. Since we used \lpt in creating the bags, we can show that~$w \geq
		\frac{1}{2}b$. To see that, consider bag~$B$ and notice that the smallest job
		in it has size at most~$\frac 12b$. When this job was assigned to its
		bag,~$B$ was a bag with the smallest size, and this size was at least~$\frac 12b$
		since we allocate jobs in \lpt-order. Hence, the size of a smallest bag
		is~$w\geq \frac{1}{2}b\geq \frac 12 T$, where the second inequality is true as
		all bags larger than~$b$ can be placed.
	
		We use this inequality to give a lower bound on the total remaining capacity on
		the~$m$ machines when the second-stage algorithm fails to place the~$(k+1)$\nobreakdash-st bag. 
		The~$(m-k)$ bags that were not placed have a combined volume of at
		least~$V_\ell = (m-k-1)w+T \geq (m-k+1)\frac{T}{2}~$. The bags that were placed
		have a combined volume of at least~$V_p = kT$. The remaining capacity is then
		at least~$C =(2-\frac{1}{m}) V_\ell + (1-\frac{1}{m})V_p$, and we have
		\begin{align*}
			C 
			 \geq \frac{2m-1}m \frac{m-k+1}2{T} + \frac{m-1}m kT  
	        \geq m T + T - \frac{m+k+1}{2m}T
	        \geq m T\,.
		\end{align*}	
		Thus, there is a machine with remaining capacity~$T$ which contradicts the 
		assumption that the bag of size~$T$ does not fit.
	\end{proof}

\section{Speed-Robust Scheduling with Equal-Size Jobs}\label{sec:Bricks}
	
	In this section, we consider instances where all jobs are of equal size (bricks) as this case seems to capture the complexity of the general problem. This intuition stems from the fact that all known lower bounds already hold for this type of instances~(see~\cite{SteinZ20} and \cref{lem:LB3:1.5}).
	By a scaling argument, we may assume that all jobs have unit processing time. Therefore, we consider unit-size jobs for the remainder of the section.

	Before
	focusing on a specific speed setting, we show that in both settings we can use
	any algorithm for infinitesimal jobs with a proper scaling to obtain a corresponding algorithm for unit-size jobs. Its
	robustness factor improves with increasing average load per bag~$n/m=:\avg$.
	Assume~$\avg>1$, as otherwise the problem is trivial.
	We define the algorithm \sandtobricks that builds on the optimal algorithm for
	infinitesimal jobs,~$\sandalg^*$, which is~$\sandalg$ for general speeds
	(Section~\ref{sec:SandOnRelated}) or~$\sandalg_{01}$  for speeds
	in~$\{0,1\}$~(Section~\ref{sec:SandOnParallel}). Let~$a_1,\ldots,a_m$ be the
	bag sizes constructed by~$\sandalg^*$ scaled such that a total processing
	volume of~$n$ can be assigned, that is,~$\sum_{i=1}^m a_i = n$. For unit-size
	jobs, we define bag sizes as~$a'_i=(1+\frac{1}{\avg}) \cdot a_i$ and assign
	the jobs greedily to the bags.
	
	\begin{lemma}\label{lem:speedScaledSand}
		For~$n$ jobs with unit processing times and~$m$ machines, 
		\sandtobricks for \srs is $(1+\frac{1}{\avg})\cdot \rho(m)$-robust, 
		where~$\rho(m)$ is the robustness factor for
		$\sandalg^*$ for~$m$~machines.
	\end{lemma}	
	\begin{proof}
		To prove the lemma, it is sufficient to show that all~$n$ unit-size jobs can be
		assigned to the constructed bags of sizes~$a'_1,\ldots,a'_m$. Suppose for the sake of contradiction that there is
		a job~$j$ that does not fit into any bag without exceeding the bag size. The
		remaining volume in the bags is at least the total capacity minus the
		processing volume of all jobs except~$j$, that is,
		\[
			\sum_{i=1}^m a'_i - (n-1) = \left(1+\frac{1}{\avg}\right) \cdot n - n +1 
			>\frac{1}{\avg} \cdot n = m\,.
		\]
		Hence, there must exist some bag that has a remaining capacity of at least~$1$ and can fit job~$j$.
	\end{proof}

\subsection{General Speeds}\label{sec:BricksOnSpeed}

	For 
	unit-size jobs, we show how to beat the robustness factor~$2-\frac1m$ (Theorem~\ref{thm:speedLPT}) for $m>4$ 	and obtain a $1.8$-robust algorithm. For~$m=2$ and~$m=3$,
	we give algorithms with best-possible robustness factors~$\frac43$ and~$\frac32$, respectively. 
	
	Theorem~\ref{thm:speedLPT} shows that LPT has a robustness factor of $2-\frac1m$, even for unit-size jobs.
	\oddalgo has a robustness factor
	increasing with the ratio between the number of jobs and the number of
	machines. \oddalgo builds bags of three possible sizes: for~$q\in \mathbb N$
	such that~$\avg\in[2q-1,2q+1]$, bags of sizes~$2q-1$ and~$2q+1$ are built,
	with possibly one additional bag of size~$2q$ (recall that $\avg=n/m$ is the average load per bag).
	\begin{restatable}{lemma}{LemSpeedLPTavg}\label{lem:speedLPTavg}
		For~$n$ unit-size jobs,~$m$ machines and~$q\in\mathbb N$ with~$\lambda\in[2q-1,2q+1]$, 
		\oddalgo is $(2-\frac{1}{q+1})$-robust for \srs.
	\end{restatable}

\begin{proof}
		The proof is along the lines of the proof of Theorem~\ref{thm:speedLPT}.
		Recall that \oddalgo builds~$m$ bags of sizes belonging to~$\{2q-1,2q,2q+1\}$,
		with at most one bag of size~$2q$. This is possible by building first~$m$ bags
		of size~$2q-1$ then putting 2 additional jobs per bag until zero or one job
		remains. Let~$m_s$ be the number of small bags (size~$2q-1$),~$m_m$ be the
		number of medium bags (size~$2q$) and~$m_b$ be the number of big bags (size
		$2q+1$). We have~$m_s+m_m+m_b=m$ and~$n = (2q+1)m_b+2qm_m+(2q-1)m_s$.
	
		We assume~$q\geq 1$ since, if~$q=0$, only bags of size~$1$ are built and the problem is trivial.
		
		Note that a small bag can be executed at speed~$q$ and a large or medium bag can
		be executed at speed~$q+1$ while respecting the prescribed makespan of~$2-\frac
		1 {q+1} = \frac{2q+1}{q+1}$. We define the {\it weight} of a small bag
		as~$q$ and the weight of a large or medium bag as~$q+1$.
		
		In the second stage, when the speed~$s_i\in\mathbb N$ is given for each machine
		$i$, associate a {\em capacity}~$ s_i$ with  each machine. Assign the bags in
		\lpt order to the machines, each bag to the least loaded machine such that the
		total {\it weight} of bags assigned to a machine does not exceed the capacity.
		The total capacity of all bins is equal to~$n$. If all bags can be assigned to
		the machines, then the total {\it size} of the bags assigned to a machine of
		speed~$s_i$ is at most~$(2-\frac{1}{q+1})s_i$, which gives the result.
		
		Suppose for the sake of contradiction that there is a bag that cannot be assigned to a
		machine. Let~$T$ be the {\it weight} of this bag. It suffices to show that the
		total remaining capacity on all machines is at least~$m(T-1)+1$. Indeed, weights
		and capacities are integers, so if the average remaining capacity per machine is
		strictly larger than~$T-1$, one machine has a remaining capacity at least~$T$ and
		the bag fits.
		
		Assume first that the bag is small, i.e.,~$T=q$. The total weight placed so far is at
		most~$(q+1)(m_b+m_m)+q(m_s-1)$, so the remaining capacity is: 
		\begin{align*}
			C&\geq n - (q+1)m_b - (q+1)m_m-qm_s+q \\
			&= (2q+1)m_b + 2q m_m +(2q-1)m_s - (q+1)m_b- (q+1)m_m-qm_s+q \\
			&= qm_b +(q-1)m_m+ (q-1)m_s +q\\
			&\geq m(T-1)+q\,.
		\end{align*}
		
		Assume now that the bag is big or medium, so~$T=q+1$. The total weight placed so far
		is at most~$(q+1)m_b$ and, thus, the remaining capacity is at least: %
		\begin{align*}
			C&\geq n - (q+1)m_b +q \\
			&= (2q+1)m_b + 2qm_m + (2q-1)m_s - (q+1)m_b +q \\
			&= qm_b+ qm_m +qm_s + (q-1) m_s+ qm_m +q\\
			&\geq m(T-1)+q\,.
		\end{align*}
		
		We conclude that all bags can be assigned to the machines without exceeding the capacity.
		Hence, this algorithm is $(2-\frac{1}{q+1})$-robust.
	\end{proof}

	Notice that the robustness guarantees in Lemmas \ref{lem:speedScaledSand} and
	\ref{lem:speedLPTavg} are functions that are decreasing in~$\avg$ and increasing in
	$\avg$, respectively. By carefully choosing between \oddalgo and \sandtobricks, depending on
	the input, we obtain an improved algorithm for bricks. For~$\avg< 8$, we
	execute \oddalgo, which yields a robustness factor of at most~$1.8$ by
	\Cref{lem:speedLPTavg}, as~$q\leq 4$ for~$\avg<8$. Otherwise, when~$\avg \geq 8$,
	we run \sandtobricks and obtain a guarantee of~$\frac{9}{8}\cdot \frac{e}{e-1}
	\approx 1.78$ by Lemma~\ref{lem:speedScaledSand}.
	\begin{theorem}\label{theo:speedbrickUB}
		There is an algorithm for \srs with unit-size jobs that has a robustness
		factor of at most~$1.8$ for any~$m\geq 1$.
	\end{theorem}

	We give a general lower bound on the best achievable robustness factor. Note
	that the lower bound of~$\rhosand$ from \Cref{thm:speedsandLB} remains valid in
	this setting and is larger than~$1.5$ for~$m\geq 6$.
	\begin{restatable}{lemma}{LemmaLBOneFive}\label{lem:LB3:1.5}
		For every~$m\geq3$, no algorithm for \srs can have a robustness factor smaller
		than~$1.5$, even restricted to unit-size jobs.
	\end{restatable}

	\begin{proof}
		Consider an instance with~$2m$ unit-size jobs. If an algorithm places~$3$ jobs
		in a bag, the adversary selects identical speeds which leads to a makespan
		$3/2$ times larger than the optimal. Otherwise, the adversary chooses a speed
		1 for~$m-1$ machines and a speed~$m+1$ for the remaining machine, thus being
		able to complete the instance within a makespan 1. The algorithm then has to
		put all the bags on the fastest machine to obtain a robustness factor smaller
		than 2. The factor is equal to~$2m/(m+1)$ which is at least~$1.5$ for~$m\geq3$.
	\end{proof}

	For special cases with few machines, we give best possible algorithms which
	match the previously mentioned lower bounds. We also show for~$m=6$ a lower
	bound larger than~$\rhosandm{6}>1.5$; proofs can be found in Appendix~\ref{apx:BricksOnSpeed}. Similar
	lower bounds have been found by a computer search for many larger values of~$m$, for
	which the difference to~$\rhosand$ tends towards zero when~$m$ grows.
	\begin{restatable}{lemma}{genSpeedsTwoThree} \label{lem:speedTwoThreeMach}
		An optimal algorithm for \srs for unit-size jobs has robustness factor~$4/3$
		on~$m=2$ machines and~$3/2$ on~$m=3$ machines, and larger than
		$\rhosandm{6}>1.5$ for~$m=6$.
	\end{restatable}

\subsection{Speeds in \texorpdfstring{$\{0,1\}$}{0,1}}\label{sec:BricksUB}

When considering speeds in $\{0,1\}$, bricks are of particular interest as the currently best known lower bound for rocks (arbitrary jobs) is~$\frac43$ and uses only bricks (unit-size jobs)~\cite{SteinZ20}.
We present
a matching upper bound. 
				
	\begin{theorem}\label{thm:BricksUB}
		There exists a $\frac 4 3$-robust algorithm for \srs with $\{0,1\}$-speeds and unit-size jobs.
	\end{theorem}

	In the proof, we handle different cases depending on~$\nrbags$
	and~$\opt{m}$ by carefully tailored methods. (Recall,~$\opt{m}$ denotes the optimal makespan on~$m$ machines.) Again, we use~$\lambda$ to denote the average load per bag, i.e.,~$\lambda := \frac nm$.
	
	For~$\opt{\nrbags}\geq 11$, we use the algorithm \sandtobricks based on \sandalgid; see \Cref{lem:bricks:Om>=11}. For~$\opt{\nrbags}\in\{9,10\}$, we refine this approach and show that for~$m\geq 40$ it is still possible to use bag sizes based on \sandalgid; see \Cref{lem:bricks:9<=Om<=10}.  
	For~$\opt{\nrbags} \in \{3,\ldots,8\}$ and~$m \geq 50$, we explicitly give a packing in \Cref{lem:bricks:3<=Om<=8} while for~$\opt{\nrbags} \in \{1,2\}$ and~$m \geq 50$ packing bags according to the optimal schedule on~$\nrbags$ machines is sufficient; see \Cref{lem:bricks:Om<=2}. 	The remaining cases,~$\opt{\nrbags}\leq 10$ and~$m\leq 50$, can be verified by enumerating all possible instances and using an integer linear program to verify that there is a solution of bag sizes that is $\frac{4}{3}$-robust; see \Cref{lem:bricks:enum}.
	
	While proving these results, inductively applying the following lemma allows us to
	restrict ourselves to instances where~$\opt{\nrbags} < \opt{\nrbags-1}$. Hence, we can express~$n = \nrbags \opt{\nrbags} - \ell$, where~$0 \leq \ell < \min\{\nrbags,\opt{\nrbags}\}$. 	
	\begin{lemma}\label{lem:m-1}
		Fix a job set~$I$ and some~$\rho$. If~$\optm=\opt{m-1}$, then any solution for~$I$
		that is $\rho$-robust for~$m-1$ bags on~$m-1$ machines is $\rho$-robust in the instance on~$m$ machines as well.
	\end{lemma}
	\begin{proof}
		Let~$I$ be the set of jobs. Compute a solution with~$m-1$ bags that is
		$\rho$-robust for the instance with jobs~$I$ and~$m-1$ bags.
		If~$\nrmach<m$ machines actually work, i.e., $m-\nrmach$ machines fail, return the schedule computed by the
		$\rho$-robust algorithm with $m-1$ bags and~$\nrmach$ machines.
		If~$m$ machines work, assign each of the~$m-1$ bags to its private machine. By assumption, the load of the largest bag is at most
		$\rho\opt{m-1}=\rho\optm$ which gives the result.
	\end{proof}

\begin{coro}\label{lem:bricks:Om>=11}
	If~$\opt{\nrbags}\geq11$, \sandtobricks based on $\sandalg_{01}$ is $\frac 4 3$-robust.
\end{coro}
\begin{proof}
	This follows directly from \Cref{lem:speedScaledSand} and the fact that,
	for~$\opt{\nrbags}\geq11$, we have~$\avg=\frac{n}{m} \geq 10$. Thus, the robustness ratio of
	\sandtobricks is at most~$(1+\frac{1}{\avg})\rhoid \leq 1.1 \cdot \rhoid <1.33$ which
	is less than~$4/3$.
\end{proof}

\begin{lemma}\label{lem:bricks:9<=Om<=10}
	For~$\opt{\nrbags}\in\{9,10\}$ and~$m\geq 50$, there is a $\frac{4}{3}$-robust algorithm.
\end{lemma}
\begin{proof}
	Consider bags created by $\sandalg_{01}$ and scale them by a factor
	of~$\smash{\frac{4}{3\rhoid}}$. We obtain bag sizes
	$
	a_i' = \frac n m \cdot\min \{ \frac{4}{3}, \ \frac{2}{3\rhoid}
	+\  \frac 4 3 \cdot \frac{i-\frac{1}{2}}m \},
	$
	for~$i\in\{1,\ldots,m\}$.
	Round down~$a_i'$ to the nearest integer and denote the
	rounded bag size by~$a_i$.
	
	The total volume of bags before rounding is~$\sum_{i=1}^m a'_i =
	\frac{4}{3\rhoid} n >n$ and is, thus, larger than the total processing volume
	of all jobs. We will show that after rounding down the bag sizes to the nearest integer, the remaining volume is still
	at least~$n$. Therefore, we can guarantee that all unit-size jobs can be
	assigned to the bags, so the robustness factor is not larger  than the
	robustness factor~$\rhosandid$ of $\sandalg_{01}$ times the scaling
	factor~$4/(3\rhoid)$, which proves the lemma.
	
	To argue that the volume that remains after rounding is at least~$n$, we show that the loss of volume due to rounding is bounded by the term~$\big(\frac
		{4}{3\rhoid}-1\big) n$.
	We do this by carefully analyzing the loss incurred on
	three different portions of bags that correspond to different parts of the profile function $\bar f(x)= 
		\frac n m \cdot \min\left\{\frac43,\ \frac{2}{3\rhoid}
		+\  \frac 4 3 \cdot x \right\}$, which, too, is obtained by scaling the profile function of $\sandalg_{01}$. Denote
	by~$\omega$ the order of~$4n$ in the additive group~$\Z_{3m^2}$, that is,
	$\omega = \min \{i\in\N \setminus\{0\} \mid i\cdot 4n \equiv 0 \mod
	3m^2\}$.~We~claim: 
	\begin{enumerate}[label = (\roman*)]
		\item On the \emph{plateau}, that is, for all bags with~$a'_i = \frac n m \cdot \frac{4}{3}$, we loose a
		volume of at most~$L_{\max} \coloneqq 1-\frac{1}{3m}$ per
		bag.\label{claim:BricksUnderSandtwoi}
		\item On the remaining portion, i.e., the \emph{slope}, we loose for any~$\omega$ consecutive bags an average
		volume of at most~$\frac{1}{2} + \frac{4n-1}{3m^2}$ per bag.
		\label{claim:BricksUnderSandtwoii}
		\item On \emph{leftover bags} of the slope, i.e., bags that remain after
		partitioning bags on the slope into~$\omega$-sized groups, we loose a total
		volume of~$\frac{3m^2}{32n} + \frac{n}{6m^2} + \frac 1 4$ additional to the
		average loss of~$\frac{1}{2} + \frac{4n-1}{3m^2}$ per bag from~(ii).
		\label{claim:BricksUnderSandtwoiii}
	\end{enumerate}
	See \Cref{fig:bricksfit} for an illustration of the different parts.
	
	\begin{figure}[tbhp]
	\centering
	\def\Pbeta{0.586}
	\def\Palpha{1}
	\def\Pm{30}
	\def\Pk{9}
	\def\Ph{0}
	\def\Pn{(\Pk*\Pm-\Ph)}

	\begin{tikzpicture}[fct/.style={very thick},yscale=0.7,xscale=0.68]
		
		\begin{axis}[
			xmin=.49, xmax=\Pm+.51, ymin=0, ymax=14,
			extra x ticks={1,2,29,30},
			xtick={0.5,1.5,...,30.5},
			extra x tick labels={$1$,\vphantom{|}$\ldots$,$\ldots$\vphantom{|},\phantom{|}$m$\phantom{|}},
			xtick pos=left,	
			xticklabels={,,},
			ytick = \empty,
			extra x tick style={major tick length = 0pt,},
			ybar, bar width=.5cm,
			samples=18,
			x=.5cm,y=.5cm,
			domain=1:18,
			]
			
		\end{axis}
	
		\begin{axis}[
			xmin=.49, xmax=\Pm+.51, ymin=0, ymax=14,
			ybar, bar width=.5cm,
			samples=18,
			x=.5cm,y=.5cm,
			domain=1:18,
			hide axis
			]
			
			\addplot+[draw=black, thick, pattern=north west lines] 
			{\Pn/\Pm*min(4/3,2/3/1.207+4/3*(x-1/2)/\Pm)};	
			\draw[thick, black, opacity = 0.3] (0.5,4.97) -- (18.03,11.95);
		\end{axis}
		
		\begin{axis}[
			xmin=.49, xmax=\Pm+.51, ymin=0, ymax=14,
			x=.5cm,y=.5cm,
			ybar,bar width=.5cm,
			samples=3,
			domain=1:3,
			hide axis
			]
			\addplot+[draw=black, thick, fill=black!25] 
			{floor(\Pn/\Pm*min(4/3,2/3/1.207+4/3*(x-1/2)/\Pm))};
		\end{axis}
		
		\begin{axis}[
			xmin=.49, xmax=\Pm+.51, ymin=0, ymax=3/2*\Pn/\Pm,
			x=.5cm,y=.5cm,
			ybar,bar width=.5cm,
			samples=floor(\Pbeta*\Pm)-2,
			domain=4:floor(\Pbeta*\Pm)+1,
			hide axis
			]
			\addplot+[draw=black, thick, fill=black!50] 
			{floor(\Pn/\Pm*min(4/3,2/3/1.207+4/3*(x-1/2)/\Pm))};
		\end{axis}
		
		\begin{axis}[
			xmin=.49, xmax=\Pm+.51, ymin=0, ymax=3/2*\Pn/\Pm,
			x=.5cm,y=.5cm,
			ybar,bar width=.5cm,clip=false,
			samples=\Pm-floor(\Pbeta*\Pm)-1,
			domain=ceil(\Pbeta*\Pm)+1:\Pm,
			hide axis
			]
			\addplot+[draw=black, thick, fill=black!75] 
			{floor(\Pn/\Pm*min(4/3,2/3/1.207+4/3*(x-1/2)/\Pm)+0.001)};
			
			\draw[thick] (30.45,12) -- (30.7,12);
			\node[]() at (31.5,12) {$\frac{4}{3}\,\frac{n}{m}~$};
			\draw[thick] (0.25,4.97) -- (0.5,4.97);
			\node[]() at (-1,4.97) {$\frac{2}{3\rhoid}\,\frac{n}{m}~$};
			
			\foreach \a/\b/\c in {3.5/14/, 8.5/14/{,dashed}, 13.5/14/{,dashed}, 18.5/14/}{
				\edef\temp{\noexpand\draw [very thick\c] (axis cs:\a,0) -- (axis cs:\a,\b);}
				\temp
			}
			\draw[thick] (8.5,14.2) -- (8.5,14.4) -- node[above] {\footnotesize$\omega = 5$} (13.5,14.4) -- (13.5,14.2);
		\end{axis}

	\end{tikzpicture}
	\caption{Bags~$a_i'$ obtained from \sandalgid by scaling (striped) and bags
		$a_i$ rounded down to the nearest integer for~$m=30$ and~$n=270$. The plateau
		is indicated in dark grey, the sloped part in mid grey, and the leftover bags in
		light grey. The continuous sloped line indicates the function~$\bar f$ with the
		appropriate scaling. Note that, here,~$\omega = 5$, and the volume of the
		bags~$a_i$ is~$288 \geq n$.} \label{fig:bricksfit}
	\end{figure}
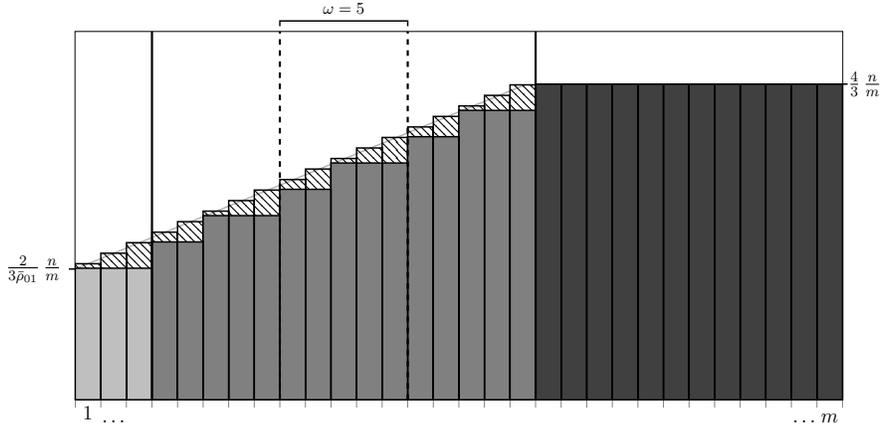

	The first case, Claim~\ref{claim:BricksUnderSandtwoi}, follows directly from
	the fact that bag sizes are~$\frac{4n}{3m}$.
	
	For Claim~\ref{claim:BricksUnderSandtwoii}, note that since the slope of the
	profile function is~$\Delta \coloneqq \frac{4n}{3m^2}$, the loss due to
	rounding is periodic. Specifically, it has a period of~$\omega$. For~$\omega$
	consecutive bags, and starting with a worst-case loss of~$L_{\max}$ for the
	first bag, the average loss due to rounding equals
	\begin{align*}
	\frac 1 \omega \sum_{i=0}^{\omega-1} \left( L_{\max} - \frac i \omega \right)
	= 1 - \frac{1}{3m^2} - \frac 1 2 \frac {\omega -1}{\omega} = 1  - \frac 1 2 -
	\frac{1}{3m^2} + \frac {1}{2\omega}.
	\end{align*}
	Denote this term by~$L_\omega$. Since~$\omega \geq \frac{3m^2}{4n}$, the
	average loss is at most~$\frac{1}{2} + \frac{4n-1}{3m^2}$.
	
	Finally, for Claim~\ref{claim:BricksUnderSandtwoiii}, we construct adversarial
	leftover bags that maximize the loss due to rounding. The first bag has a loss
	of~$L_{\max}$, the second a loss of~$L_{\max}-\Delta$, and so on until the
	last one has a loss which is just above~$L_\omega$. Adding further leftover
	bags would lead to averaging as in Claim~\ref{claim:BricksUnderSandtwoii} and
	ultimately a smaller rounding loss. For such adversarial leftover bags, we
	define~$n_\ell \coloneqq \frac{L_{\max} - L_\omega}\Delta$, so~$\lfloor n_\ell
	\rfloor$ is the number of leftover bags. Then, the overall additional loss,
	when compared to~$L_\omega$, is given by
	\begin{align*}
	\sum_{i=0}^{\floor{n_\ell}-1} \left( (L_{\max} - L_\omega) - i \cdot \Delta \right) 
	&\leq \floor{n_\ell}(L_{\max} - L_\omega) - \frac{\floor{n_\ell}(n_\ell-2)}{2}\cdot \Delta\\
	\leq \frac {n_\ell\,(L_{\max} - L_\omega)}{2} + n_\ell\cdot\Delta
	&\leq \frac {(\frac 1 2 - \frac {1}{2\omega})^2}{2\Delta} + \frac 1 2.
	\end{align*}
	Again, we use~$\omega \geq \frac{3m^2}{4n}$ to obtain an upper bound for the last term. Then
	\begin{align*}
	\sum_{i=0}^{\floor{n_\ell}-1} \left( (L_{\max} - L_\omega) - i \cdot \Delta\right) 
	& \leq \frac {(\frac 1 2 - \frac {2n}{3m^2})^2}{\frac{8n}{3m^2}} + \frac{1}{2}
	= \frac{3m^2}{32n} + \frac{n}{6m^2} + \frac 1 4,
	\end{align*}
	which concludes the proof of Claim~\ref{claim:BricksUnderSandtwoiii}.
	
	Let~$[x]$ denote the value of~$x$ rounded to the closest integer. We can now
	bound the overall loss due to rounding by
	\begin{align*}
	\left[(1-\beta)m\right] \cdot L_{\max} + \left[\beta m\right] \cdot L_\omega + \left(\frac{3m^2}{32n} + \frac{n}{6m^2} + \frac 1 4\right).
	\end{align*}
	Using $(1-\beta)m$ and $\beta m$ instead of rounding to the nearest integer decreases this term by at most $\tfrac{1}{2}$ since $m$ is an integer and $L_{\max},L_\omega \leq 1$.
	Thus the total loss due to rounding is less or equal to~$\big(\frac {4}{3\rhoid}-1\big) n$, if 
	\begin{align*}
	\left(\frac {4}{3\rhoid}-1\right) n - 
	(1-\beta)m -
	\beta m \cdot \left(\frac{1}{2} + \frac{4n-1}{3m^2}\right)
	-\frac{3m^2}{32n} - \frac{n}{6m^2} - \frac 1 4 &\geq 0 ,
	\end{align*}
	Using~$\nrbags \opt{\nrbags}   - \opt{\nrbags} \leq n \leq \nrbags \opt{\nrbags} $, 
	algebraic computations show that for~$\opt{\nrbags}=9$ and~$\opt{\nrbags}=10$ this is the case
	when~$m\geq40$ and~$m\geq30$, respectively (see \cite{ehmnss2021}).
\end{proof}

When $\opt{\nrbags} \leq 8$, we only require constantly many different bag sizes, which we describe explicitly. To simplify the analysis, we assume that, once the number~$\nrmach$ of non-failing machines is revealed, bags are assigned to the machines in \lpt order. Since this assignment cannot be better than the optimal bag-to-machine assignment, the ratio between the makespan attained by \lpt and the optimal solution is not smaller than the robustness factor. 

The following observation on bags assigned to machines by \lpt is crucial. For a given set of bags, let~$\lptf$ be the makespan attained by assigning the bags in \lpt order to the currently least loaded machine 
when there are~$\nrmach$ machines with speed~$s_i = 1$ (and $\nrbags - \nrmach$ machines with speed~$s_i =0$).
\begin{lemma}\label{lem:bricks:lpt}
	Let~$a$ be the size of a bag determining the makespan of \lpt on~$\nrmach$
	machines. If~$a \leq \frac{\opt{\nrmach}}{3}$, then~$\lptf \leq \frac43
	\opt{\nrmach}$.
\end{lemma}
\begin{proof}
	Let~$b$ be a bag of size~$a$ that determines the makespan and is placed on machine~$i$. As \lpt assigns bags
	in decreasing size to the currently least loaded machine, the load
	on machine~$i$ right before assigning~$b$ was at most~$\opt{\nrmach}$.
	Hence, it follows that~$\lptf \leq \opt{\nrmach} + a \leq \frac{4}{3} \opt{\nrmach}$.
\end{proof}

For~$\opt{\nrbags} \in \{3,\ldots,8\}$, we pack four different types
of bags depending on~$\opt{\nrbags}$. 
For~$l \in \{0,1,2,3\}$, we denote
by~$a_l$ the size of the~$l$-th bag type and by~$x_l$ its multiplicity. The idea
is to have~$x_0 + x_1 = x_3 \approx \frac25 m$ while~$x_2 \approx \frac15m$.

More precisely, let~$a_1 = \ceil{\frac23\opt{\nrbags}}$,~$a_2 = \opt{\nrbags}$, and~$a_3 =
\floor{\frac43\opt{\nrbags}}$ be the three \emph{standard} bag sizes. Since~$a_1 + a_3 = 2
a_2$, packing as many smallest as largest bags, i.e.,~$x_1 = x_3$, ensures
that $\sum_{l=1}^3 a_l x_l = \nrbags \opt{\nrbags}$. Recall that $n = \nrbags \opt{\nrbags} - \ell$
with~$0 \leq \ell < \opt{\nrbags}$. Hence, we decrease~$x_1$ by~$\ell$ and pack~$x_0 =
\ell$ many bags of size~$a_0 = a_1 - 1 = \ceil{\frac23\opt{\nrbags}} -1$ in order to pack
exactly~$n$ jobs in our~$\nrbags$ bags. As we aim for a tight robustness
guarantee, we have to be careful about the exact number of bags in this
section. The following table defines~$x_0+x_1$,~$x_2$, and~$x_3$ depending on~$m
\Mod 5$.
\begin{equation*}
\renewcommand*\arraystretch{1.4}
\arraycolsep=3pt
\begin{array}{c||ccccc|c}
m \Mod 5 & 0 & 1 & 2 & 3 & 4 & a_l\\\hline
x_0 \!+\! x_1 & \frac{2m}5 & 2 \floor{\frac m5} & 2 \floor{\frac m5} \!+\! 1 & 2 \floor{\frac m5} \!+\! 1 & 2 \floor{\frac m5} \!+\! 1 & \ceil{\frac23 \opt{\nrbags}}\\
x_2 & \frac m5 & \floor{\frac m5} \!+\! 1 & \floor{\frac m5} & \floor{\frac m5} \!+\! 1 & \floor{\frac m5} \!+\! 2 & \opt{\nrbags}\\
x_3  & \frac{2m}5 & 2 \floor{\frac m5} & 2 \floor{\frac m5} \!+\! 1 & 2 \floor{\frac m5} \!+\! 1 & 2 \floor{\frac m5} \!+\! 1 & \floor{\frac43\opt{\nrbags}}
\end{array}
\end{equation*}

The analysis of assigning these bags via \lpt is a tedious case distinction based on which types of bags are assigned to the same machine and does not provide new insights. Therefore, we defer it to Appendix~\ref{apx:BricksUB}. 

\begin{restatable}{lemma}{bricksforthreeOmeight}\label{lem:bricks:3<=Om<=8}
	If~$\opt{\nrbags} \in \{3,\ldots,8\}$ and~$\nrbags \geq 50$, there is a $\frac43$-robust algorithm. 
\end{restatable}

\begin{lemma}\label{lem:bricks:Om<=2}
	If~$\opt{\nrbags} \leq 2$ and~$\nrbags \geq 37$, there exists a $\frac43$-robust algorithm. 
\end{lemma}
\begin{proof}
	If~$\opt{\nrbags} = 1$, each job gets a unique bag. Hence,~$\lptf = \opt{\nrmach}$. 
	
	For~$\opt{\nrbags} = 2$, if no machine fails, the packing achieves a makespan of~$\opt{\nrbags}$. By
	\cref{lem:m-1}, we may assume that~$n\in \{2m,2m-1\}$. Therefore,
	if~$\frac{\nrbags}{2} \leq \nrmach \leq \nrbags-1$, we have~$\opt{\nrmach}
	\geq 3$ while~$\lptf \leq 4$ since at most~$2$ bags are assigned to the same
	machine. For the remaining cases, with~$\nrmach \leq \frac{\nrbags}{2}-1$, we
	have~$\opt{\nrmach} \geq 5$. When~$\opt{\nrmach} = 5$, \lpt assigns at
	most~$3$ bags to each machine which guarantees that~$\lptf \leq 6 \leq \frac43
	\opt{\nrmach}$. If~$\opt{\nrmach} \geq 6$, \cref{lem:bricks:lpt} implies
	that~$\lptf \leq \frac43\opt{\nrmach}$.
\end{proof}
\begin{restatable}{lemma}{BricksEnum}\label{lem:bricks:enum}
	For~$\opt{\nrbags}\leq 10$ and~$m\leq 50$, there exists a $\frac{4}{3}$-robust algorithm.
\end{restatable}
\begin{proof}
	We enumerate all instances with~$\opt{\nrbags}\neq\opt{m-1}$ for which the properties of the
	statement hold. These instances consist
	of~$n=\nrbags \opt{\nrbags} -\ell$ unit-size jobs, with~$\opt{\nrbags}\in\{1,2,\ldots,10\}$
	and~$\ell\in\{0,1,\ldots,\min(\opt{\nrbags}-1,m-1)\}$. We solve each such instance by
	computing a feasible solution for an integer linear program (ILP) which we now
	describe. Denote by~$p_{\max}$ the largest possible size of a bag. To be able
	to achieve a robustness of~$\frac{4}{3}$, we must have~$p_{\max}=
	\floor{\frac{4}{3}\opt{\nrbags}}$. The decision variables used are~$y_p$
	and~$x_{p,i,\nrmach}$ to indicate how many bags of size~$p$ are created and
	how many of these bags of size~$p$ are assigned to machine~$i$ when~$\nrmach$
	machines are working, respectively. The ILP is as follows. 
	\begin{equation*}\label{eq:mmdk:ilp}
	\begin{array}{lcll} 
	\textstyle{\sum_{p=1}^{p_{\max}}}\; y_{p} & = & \nrbags & \\[5 pt] 
	\textstyle{\sum_{p=1}^{p_{\max}}}\; p\cdot y_{p} & = & n &  \\[5 pt] 
	\textstyle{\sum_{i = 1}^{\nrmach}}\; x_{p,i,\nrmach} & = & y_p & \text{~for all } \nrmach \in [\nrbags], p \in [p_{\max}]\\[5 pt] 
	\textstyle{\sum_{p=1}^{p_{\max}}}\; p\cdot x_{p,i,\nrmach} & \leq & \frac 4 3 \cdot\optf & \text{~for all } \nrmach \in [\nrbags], i \in [\nrmach]\\[5 pt]
	x_{p,i,\nrmach} & \in & \Z_{\geq 0} & \text{~for all } p \in [p_{\max}], \nrmach \in [\nrbags], i \in [\nrmach]\\[5 pt] 
	y_{p} & \in & \Z_{\geq 0} & \text{~for all } p \in [p_{\max}]
	\end{array}
	\end{equation*}
	The first equation ensures that exactly~$m$ bags are created, and the second
	that, in total, they consist of exactly~$n$ jobs. The third equation enforces
	that, for every value of~$\nrmach$, all bags are assigned to a machine.
	Finally, the fourth equation checks that, for every value of~$\nrmach$ and for
	every remaining machine, the makespan of the optimal solution is not exceeded
	by more than a factor of~$\frac 4 3$, yielding the robustness guarantee. The
	last two constraints dictate integrality and non-negativity of the decision
	variables.
	
	A file containing the solutions produced by this ILP and a simple program
	verifying them are available in \cite{ehmnss2021}.
\end{proof}

\section{Concluding Remarks}	

	In this work, we have
	established matching lower and upper bounds
	for the \srs problem with infinitesimal jobs and design optimal algorithms for either infinitesimal jobs or equal-size
	jobs when
	speeds are restricted to~$\{0,1\}$. We believe that the insights from our optimal algorithms will be useful
	to improve the more general~upper~bounds.
	
	We have also shown that randomization does not help when the speeds belong to
	$\{0,1\}$ and jobs are infinitesimal. However, the other known lower bounds do
	not hold in a randomized setting, so designing better randomized algorithms
	remains an interesting challenge.
	
	The following observation about adversarial strategies might be
	useful for further research. We give two somewhat orthogonal examples proving
	the lower bound of~$\frac 43$ for \srs with unit processing time jobs and
	speeds from~$\{0,1\}$. In both examples, there are only two relevant
	adversarial strategies: either one machine fails or none. This may seem
	sub-optimal, but the lower bound of~$\frac 43$ is tight for unit-size jobs
	(\Cref{thm:BricksUB}). Further, we show in the proof of 
	\Cref{thm:IdenticalWSand} (\Cref{apx:SandOnIdentical}) that, for infinitesimal
	jobs, an adversary only requires two strategies to force all algorithms to have
	a robustness factor at least~$\rhosandid(m)$, which is optimal.
	
	{\em Example 1} (from \cite{SteinZ20}). Consider~$2m$ jobs and~$m>2$
	machines. If an algorithm places~$2$ jobs per bag, let one machine fail. This
	leads to a makespan of~$4$ while the optimal makespan is~$3$ which gives a lower
	bound of~$\frac 43$. Otherwise, one bag has at least three jobs, and, if no
	machine fails, the algorithm's makespan is~$3$ while the optimal makespan is~$2$,
	yielding a lower bound of~$\frac 32$. 
	
	{\em Example 2.} Our new dual example has~$3m$ jobs for~$m>3$ machines.
	If an algorithm places~$3$ jobs per bag, let one machine fail. This leads to a
	makespan of~$6$ while the optimal makespan is~$4$, implying a lower bound of
	$\frac 32$. Otherwise, one bag has at least~$4$ jobs, and if no machine fails,
	the algorithm's makespan is~$4$ whereas the optimal makespan is~$3$, which again
	gives a lower bound~of~$\frac 43$.

We conclude with a remark on a natural generalization of our model in which we allow to pack more bags than the number of machines, say $M\geq m$. Our upper bounds obviously still hold by using the algorithms with $m$ bags as presented. One would expect that the approximation ratio improves when~$M/m$ increases, and it would be interesting to quantify the achievable robustness factor in terms of both $M$ and $m$. Indeed, for the general problem with rocks and arbitrary speeds, it is not difficult to see that \Cref{thm:speedLPT} can be generalized by parameterization and adjusting the inequalities in the proof (see Appendix~\ref{app:extensions}). Packing $M\geq m$ bags by LPT yields a robustness factor of $1+\frac{m-1}{M}$, which interpolates nicely between $2-\frac{1}{m}$ (for $M=m$) and~$1$~(for~$M\rightarrow\infty$).

Concerning our lower bounds, we expect \Cref{thm:speedsandLB,thm:BricksLBrand} to become invalid for $M>m$ as these results are tight for $M=m$ and an algorithm should be able,
with infinitesimal jobs, to exploit any additional bag to reduce the competitive
ratio. However, the lower bound of $4/3$ for unit-time jobs and speeds in
$\{0,1\}$ (both examples above) has some slack and therefore
holds when $M$ is moderately larger than $m$. Specifically, assuming for simplicity that $3$ divides $m$, one can easily check that for an instance with $M=\frac43m-1$ bags and  $2m$ unit
jobs, no algorithm can be better than $4/3$-competitive (the worst cases being $0$ and $\frac m3$ machine failures).
Quantifying exactly how all results in this paper evolve when $M>m$ remains an open problem.

\bibliographystyle{abbrvnat}
\bibliography{bib}   

\newpage
\appendix

\section*{Appendices}

\section{Proofs for Section~\ref{sec:SandOnParallel} -- Infinitesimal Jobs and Speeds in \texorpdfstring{$\{0,1\}$}{0,1}}\label{apx:SandOnIdentical}

In the following, we establish the upper bound of $\rhoid(m)$ when speeds are in $\{0,1\}$, which dominates the more coarse bound of $\rhoid$ from Theorem~\ref{prop:identical-steps} but also requires substantially more work.

\IdenticalWSand*

\begin{proof}
	We give bag sizes that guarantee a robustness factor of~$\rhoid(m)$ for every~$m$. While the load
	distribution in the limit approaches that given by~$\bar f$, we do not work
	with~$\bar f$ explicitly anymore. We fix~$m\geq 3$ in the following; the other
	cases are trivial. Furthermore, let again
	\[
	t^\star\in\argmax_{t\leq \frac m2,~t\in\mathbb N}\ \frac{1}{\frac t{m-t} 
		+\frac{m-2t}m}\,.
	\]
	To show the theorem, we distinguish two cases:~$t^\star \leq 1$ and~$t^\star >1$.
	
	We start with~$t^\star\leq 1$ and show that this implies $m\in \{3,4,5\}$. First, note that the above expression defining
	$t^\star$ equals $1$ for $t=0$ and $t=m/2$ and is larger for values of $t$ in
	between. As we assumed~$m \geq3$, we have that~$m/2 \geq 1$, and thus~$t^\star
	\geq 1$. Hence, for~$m \in \{3,4\}$, we have~$t^\star =1$.
	
	For~$m \geq 5$, in order to have~$t^\star =1$, it is necessary
	that~$\smash{\frac{1}{\frac t{m-t} +\frac{m-2t}m}}$ is larger at~$t =1$ than
	at~$t=2$. Noting that
	\[ 
	\frac{1}{\frac{t}{m-t}+\frac{m-2t}{m}} = \frac{m^2-tm}{m^2-2tm+2t^2}
	\]
	and reformulating, we obtain the necessary condition
	\[
	\frac{m^2-2m}{m^2-4m+8} - \frac{m^2-m}{m^2-2m+2}  {}<{} 0\,. 
	\] 
	The left-hand side is equal to 	
	\begin{align*}
		\frac{m(m-2)}{(m-2)^2+4} - \frac{m^2-m}{m(m-2)+2}  
		= \frac{m(m^2-6m+4)}{((m-2)^2+4) \cdot (m(m-2)+2)}\,.
	\end{align*}
	As~$m \geq 5$, this term is negative if and only if~$m(m^2 -6m + 4) {}<{} 0$. Since
	the roots of this polynomial are $0$, $3-\sqrt 5$, and $3+\sqrt 5$, the
	expression is strictly negative if~$m\leq 5$ and strictly positive for~$m \geq 6$, which implies~$t^\star \geq 2$ in the latter case.
	
	It remains to consider~$m=5$. By the above calculations we know that~$t^\star =
	1$. Therefore,~$t^\star \leq 1$ implies~$m \in \{3,4,5\}$. For these three
	cases, the optimal bag sizes are as follows:
	\begin{itemize}
		\item For $m=3$, the bag sizes are $\{0.9,0.9,1.2\}$ and $\rhosandidm{3}=1.2$.
		\item For $m=4$, the bag sizes are $\{0.8,0.8,1.2,1.2\}$  and $\rhosandidm{4}=1.2$.
		\item For $m=5$, the bag sizes are $\{\frac{25}{34},\frac{25}{34},\frac{40}{34},\frac{40}{34},\frac{40}{34}\}$ and $\rhosandidm{5}=40/34$.
	\end{itemize} 
	It can be easily verified that, if at most~$m/2$ machines fail, the obtained makespan is at most~$\rhosandidm{m}$.
	
	Consider now the case that~$t^\star \geq 2$. By the discussion above, this
	implies~$m \geq 6$. We also claim that there exist bag sizes that achieve a
	robustness factor of~$\rhoid(m)$. It turns out that, for many different values of $m$,
	these sizes are not unique. We impose additional constraints on the bag sizes
	to get bag sizes that are easier to analyze, giving some intuition along
	the way for why the imposed restrictions do not remove all bag sizes
	achieving~$\rhoid(m)$. When we are left with a single degree of freedom, we
	impose lower and upper bounds on the corresponding variable so that fulfilling
	these bounds implies the robustness guarantee~$\rhoid(m)$. We then show that the largest lower
	bound does not exceed the smallest upper bound, implying that there is a
	feasible choice for said variable.
	
\begin{figure}
		\centering
		\begin{tikzpicture}[fct/.style={very thick, domain=0:1.6}, samples=100, font=\small, scale=0.8]
			\begin{scope}[xscale=7.007,yscale=5.564]
				\draw[fill=black!25] (0,0) rectangle (0.05,0.534128);
				\draw[fill=black!25] (0.05,0) rectangle (0.1,0.534128);	
				\draw[fill=black!25] (0.1,0) rectangle (0.15,0.610432);
				\draw[fill=black!25] (0.15,0) rectangle (0.2,0.610432);	
				\draw[fill=black!25] (0.2,0) rectangle (0.25,0.686736);
				\draw[fill=black!25] (0.25,0) rectangle (0.3,0.686736);	
				\draw[fill=black!25] (0.3,0) rectangle (0.35,0.76304);
				\draw[fill=black!25] (0.35,0) rectangle (0.4,0.76304);	
				\draw[fill=black!25] (0.4,0) rectangle (0.45,0.839344);
				\draw[fill=black!25] (0.45,0) rectangle (0.5,0.839344);	
				\draw[fill=black!25] (0.5,0) rectangle (0.55,0.915648);
				\draw[fill=black!25] (0.55,0) rectangle (0.6,0.915648);	
				\draw[fill=black!25] (0.6,0) rectangle (0.65,0.991952);
				\draw[fill=black!25] (0.65,0) rectangle (0.7,0.991952);	
				\draw[fill=black!25] (0.7,0) rectangle (0.75,1.06826);
				\draw[fill=black!25] (0.75,0) rectangle (0.8,1.06826);
				\draw[fill=black!25] (0.8,0) rectangle (0.85,1.14456);
				\draw[fill=black!25] (0.85,0) rectangle (0.9,1.14456);						
				\draw[fill=black!25] (0.9,0) rectangle (0.95,1.20656);
				\draw[fill=black!25] (0.95,0) rectangle (1,1.20656);	
				\draw[fill=black!25] (1.0,0) rectangle (1.05,1.20656);
				\draw[fill=black!25] (1.05,0) rectangle (1.1,1.20656);	
				\draw[fill=black!25] (1.1,0) rectangle (1.15,1.20656);
				\draw[fill=black!25] (1.15,0) rectangle (1.2,1.20656);	
				\draw[fill=black!25] (1.2,0) rectangle (1.25,1.20656);
				\draw[fill=black!25] (1.25,0) rectangle (1.3,1.20656);	
				\draw[fill=black!25] (1.3,0) rectangle (1.35,1.20656);
				\draw[fill=black!25] (1.35,0) rectangle (1.4,1.20656);	
				\draw[fill=black!25] (1.4,0) rectangle (1.45,1.20656);
				\draw[fill=black!25] (1.45,0) rectangle (1.5,1.20656);	
				\draw[fill=black!25] (1.5,0) rectangle (1.55,1.20656);
				\draw[fill=black!25] (1.55,0) rectangle (1.6,1.20656);
				
				\draw (0.075,0.610432) -- (0.1,0.610432);
				\draw[<->] (0.0875,0.610432) -- (0.0875,0.534128);
				\node at (0.0625,0.57228) {\tiny $\delta$};
				
				\draw (0.175,0.686736) -- (0.2,0.686736);	
				\draw[<->] (0.1875,0.686736) -- (0.1875,0.610432);
				\node at (0.1625,0.648584) {\tiny $\delta$};
				
				\draw (0.275,0.76304) -- (0.3,0.76304);	
				\draw[<->] (0.2875,0.76304) -- (0.2875,0.686736);
				\node at (0.2625,0.724888) {\tiny $\delta$};
				
				\draw (0.375,0.839344) -- (0.4,0.839344);	
				\draw[<->] (0.3875,0.839344) -- (0.3875,0.76304);
				\node at (0.3625,0.801192) {\tiny $\delta$};
				
				\draw (0.475,0.915648) -- (0.5,0.915648);	
				\draw[<->] (0.4875,0.915648) -- (0.4875,0.839344);
				\node at (0.4625,0.877496) {\tiny $\delta$};
				
				\draw (0.575,0.991952) -- (0.6,0.991952);	
				\draw[<->] (0.5875,0.991952) -- (0.5875,0.915648);
				\node at (0.5625,0.9538) {\tiny $\delta$};
				
				\draw (0.675,1.06826) -- (0.7,1.06826);	
				\draw[<->] (0.6875,1.06826) -- (0.6875,0.991952);
				\node at (0.6625,1.030106) {\tiny $\delta$};

				\draw (0.775,1.14456) -- (0.8,1.14456);	
				\draw[<->] (0.7875,1.14456) -- (0.7875,1.06826);
				\node at (0.7625,1.10641) {\tiny $\delta$};

				\draw (0.875,1.20656) -- (0.9,1.20656);	
				\draw[<->] (0.8875,1.20656) -- (0.8875,1.14456);
				\node at (0.8575,1.17556) {\tiny $\delta'$};
				
				\node at (0.025,0.267) {\tiny $a_{\scaleto{1}{2pt}}$};
				\draw[->] (0.025,0.307) -- (0.025,0.534128);
				\draw[->] (0.025,0.227) -- (0.025,0);	
				
				\node at (0.075,0.267) {\tiny $a_{\scaleto{2}{2pt}}$};
				\draw[->] (0.075,0.307) -- (0.075,0.534128);
				\draw[->] (0.075,0.227) -- (0.075,0);
				
				\node at (0.125,0.3052) {\tiny $a_{\scaleto{3}{2pt}}$};
				\draw[->] (0.125,0.3352) -- (0.125,0.610432);
				\draw[->] (0.125,0.2752) -- (0.125,0);	
				
				\node at (0.175,0.3052) {\tiny $a_{\scaleto{4}{2pt}}$};
				\draw[->] (0.175,0.3352) -- (0.175,0.610432);
				\draw[->] (0.175,0.2752) -- (0.175,0);				
				
				\node at (0.225,0.3434) {\scaleto{$...$}{0.5pt}};																					
			\end{scope}					
			\begin{axis}[	
				axis lines=middle,
				axis line style={-},
				ylabel near ticks,
				xlabel near ticks,
				ylabel = {bag size},
				xlabel = {bags},
				xtick = \empty,
				ytick = \empty,
				extra x ticks = {0.025,0.875,1.575},
				extra x tick labels = {1,$2t^\star=18$,$m=32$},			
				extra y ticks={0.33333,0.66666,1,1.2067,1.33333},
				extra y tick labels={$1/3$,$2/3$,$1$,$\bar{\rho}_{01}(32)$,$4/3$},
				domain=0:1.6, ymin=0, xmax=1.6, xmin=0, ymax=1.33333,
				width=12.8cm, height=9cm]
				\addplot[domain=0:2.36, gray, densely dotted] (1.6,x);
				\addplot[domain=0:0.8875, gray, densely dotted] (x,1.2067);	
			\end{axis}
		\end{tikzpicture}
		\caption{The structure of the optimal solution of Theorem~\ref{thm:IdenticalWSand} illustrated for $m=32$.}\label{fig:sol-structure}
	\end{figure}
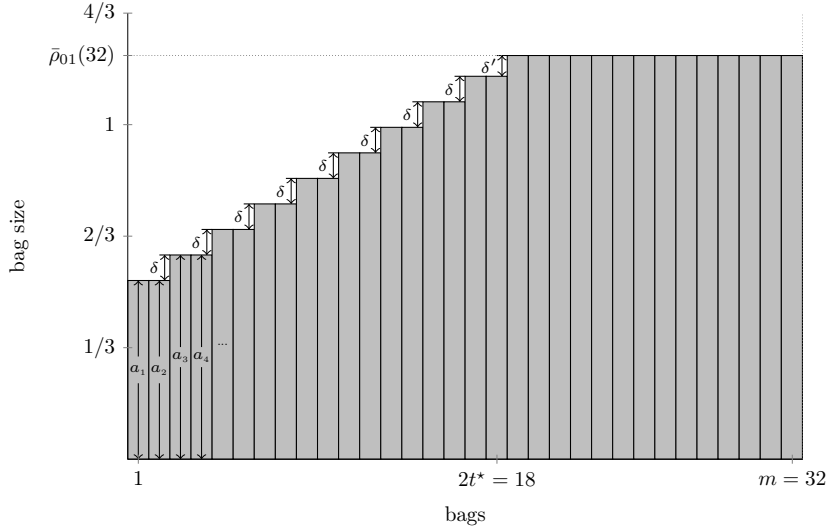	
	
	As observed in Section~\ref{sec:SandOnParallel}, if we want to guarantee a robustness factor of~$\rhoid(m)$, we need to fix
	$a_{2t^\star+1},\dots,a_m=\rhosandid(m)$; so, this restriction
	does not remove any set of
	bag sizes that achieves~$\rhoid(m)$. Since the optimal factor could be achieved by
	a set of bag sizes that achieves this factor with the 
	simple folding, we also assume such bag sizes here. As before, this implies
	$a_{i}+a_{2t^\star+1-i}=\rhosandid(m)\cdot m/(m-t^\star)$ for all~$i=1,\dots,t^\star$.
	
	Note that, when considering a bag~$i\leq t$
	for any $t<m/2$ and keeping
	track of the bag which~$i$ is matched up with while increasing~$t$,
	odd bags are always folded on top of even bags and vice versa.
	That motivates considering bags in pairs and thinking of pairs that get matched
	up rather than particular bags. Specifically, we choose~$a_{2i-1}=a_{2i}$
	for~$i=1,\dots,t^\star$.
	
	To imitate the increasing linear part of~$\bar f$, we further
	impose~$a_{2i+1}=a_{2i-1}+\delta$ for some~$\delta\geq 0$ and
	all~$i=1,\dots,t^\star-1$. Since $\bar f$ has two different slopes, similar to
	our previous set of bag sizes, the difference between~$\rhoid(m)$ and the
	largest bag that is smaller than~$\rhoid(m)$ may differ from~$\delta$; we
	call this value~$\delta':=a_{2t^\star+1}-a_{2t^\star}$. In fact, it can be
	shown that~$\delta '\neq \delta$ is even necessary given our previous
	assumption on~$a_i$ for some values of~$m$. 
	Figure~\ref{fig:sol-structure} visualizes this structure for $m=32$. 
		Note that Figure~\ref{fig:folding} visualizes folding on the structure in the $m=20$ case.
	
	Note that there is a single degree of freedom left if we want to define bags of
	total volume precisely~$m$: If we choose~$\delta$, then this fixes~$\delta'$, and
	vice versa. The goal is to show that the set that we can choose~$\delta$ from,
	so as to guarantee~$\rhoid(m)$, is nonempty. In what follows, we derive lower and
	upper bounds on~$\delta$. Fulfilling these bounds implies a robustness factor of~$\rhoid(m)$.
	
	To make our computations simpler, we assume that~$\delta\geq\delta'$. This
	yields our first lower bound~$L_1$ on~$\delta$. We express all bounds in terms
	of~$m$, $t^\star$, and~$\rhoid(m)$. To compute~$\delta'$ from these values,
	observe that the total size of bags~$a_1,\dots,a_{2t^\star}$
	is~$m-\rhoid(m)\cdot(m-2t^\star)$ using that the total volume is~$m$ and that
	bags~$a_{2t^\star+1},\ldots, a_\nrbags$ have size~$\rhoid(m)$. Dividing
	by~$2t^\star$ yields the average bag size of bags~$a_1,\dots,a_{2t^\star}$,
	denoted by
	\[
	\bar a=\frac{m-\rhoid(m)\cdot(m-2t^\star)}{2t^\star}\,.
	\]
	From~$\bar a$ it takes~$(t^\star-1)/2$ steps of size $\delta$ to get to~$a_{2t^\star}$. Arguing
	in terms of volume, which allows arguing in terms of half steps, we get: If
	$\delta'\le\delta$, then we have~$\rhoid(m)= a_{2t^\star}+\delta'\le \bar a +
	\delta(t^\star-1)/2 +\delta= \bar a + \delta(t^\star+1)/2.$ This is equivalent
	to
	\begin{align}
	\delta \  \ge\  \frac{2}{t^\star +1} \cdot \left(\rhoid(m)- \bar a\right) 
	\ &= \  \frac{2}{t^\star +1} \cdot \left(\rhoid(m)- \frac{m-\rhoid(m)\cdot(m-2t^\star)}{2t^\star}\right) \nonumber\\
	\ &=\  \frac{m(\rhoid(m)-1)}{t^\star(t^\star+1)}=:L_1\,.\label{eq:qm-lb1}  
	\end{align}
	Now, we give two bounds on~$\delta$ that ensure that the bag sizes are feasible
	in that we have $a_i\in[0,\rhoid(m)]$ for all~$i$. Using
	that~$a_{2t^\star}=\bar{a}+(t^\star-1)/2\cdot\delta$ and imposing
	that~$a_{2t^\star}\leq\rhoid(m)$, i.e., that
	$\delta'=a_{2t^\star+1}-a_{2t^\star}\ge0$, yields $\rhoid(m)\geq \bar a +
	(t^\star -1)\delta/2$. Hence,
	\begin{align}
	\delta \ \leq\  \frac{2}{t^\star-1}\cdot (\rhoid(m)-\bar a) &= 
	\frac{2}{t^\star-1}\cdot \left(\rhoid(m)-\frac{m-\rhoid(m)\cdot(m-2t^\star)}{2t^\star}\right) \nonumber\\
	&=\  \frac{m(\rhoid(m)-1)}{t^\star(t^\star-1)}=:U_1\,.\label{eq:qm-ub1}
	\end{align}
	We also need~$a_1\geq 0$, which we impose by letting the sum of the 
	increments not exceed~$\rhoid(m)$. This yields an upper bound  on $\delta$ of
	\begin{equation}
	U_2:= \frac{\rhoid(m)}{t^\star-1}\,.\label{eq:qm-ub2}
	\end{equation}

	For the upper and lower bounds ensuring the robustness, recall
	that~$a_{i}+a_{2t^\star+1-i}=\rhoid(m)\cdot \frac{m}{m-t^\star}$ for all
	$i=1,\dots,t^\star$. That is, the robustness ratio is attained \emph{exactly}
	when $t=t^\star$. This implies that the ratio of the increase in the
	algorithm's cost and the increase the optimum's cost when the number~$t$ of
	failing machines is increased from $t^\star$ to, say, $t^\star+k$, should be at
	most $\rhoid(m)$. To compute the change in the algorithm's cost, note that,
	using~\eqref{eq:qm-lb1}, \emph{any} bag size~$a_i$ can be bounded from above by
	$a_1+\lfloor i/2\rfloor\cdot \delta$. That implies that the algorithm's cost
	increases by at most~$k\delta$. On the other hand the optimum's cost changes
	from $\frac{m}{m-t^\star}$ to $\frac{m}{m-t^\star-k}$. So the aforementioned
	ratio is
	\[
	\frac{k\cdot \delta}{\frac{m}{m-t^\star-k}-\frac{m}{m-t^\star}}=\frac{\delta\cdot(m-t^\star)(m-t^\star-k)}{m}\,,
	\]
	yielding an upper bound on $\delta$ of
	\[
	U^k_3:=\frac{m\cdot \rhoid(m)}{(m-t^\star)(m-t^\star-k)}\,.
	\]
	This upper bound is minimized for $k=1$, so we only need to remember
	$U_3:=U_3^1$. Similarly, we consider the case when $t$ is decreased from
	$t^\star$ to $t^\star-k$. Then the algorithm's cost decreases by
	\emph{precisely} $k\delta$ and that of the optimum by
	$\frac{m}{m-t^\star}-\frac{m}{m-t^\star+k}$. Imposing that the ratio of these
	quantities is \emph{at least} $\rhoid(m)$ yields
	\[
	\delta\geq\frac{m\cdot \rhoid(m)}{(m-t^\star)(m-t^\star+k)}=:L_2^k\,.
	\] 
	Again, this bound is \emph{maximized} for $k=1$, so we only need to remember $L_2:=L_2^1$.
	
	With these lower and upper bounds on $\delta$, we can complete the proof. It
	boils down to showing that the interval~$[\max\{L_1,L_2\},\min\{U_1,U_2,U_3\}]$
	is nonempty because then we can choose $\delta$ from that interval and thereby
	define bag sizes with a robustness of $\rhoid$. To achieve this, we first
	rewrite~$\rhoid(m)$ as
	\begin{equation}
	\rhoid(\nrbags) = \frac{1}{\frac {t^\star}{m-t^\star}+\frac{m-2t^\star}m} =
	\frac{m(m-t^\star)}{m^2-2mt^\star+2(t^\star)^2}\,.\label{eq:rhoidexpanded}
	\end{equation}
	We use \eqref{eq:rhoidexpanded} to expand the lower and upper bounds as follows:
	\begin{align*}
		L_1 &= \frac{m(\rhoid(m)-1)}{t^\star(t^\star+1)} 
		=\frac{m(m-2t^\star)}{(t^\star+1)(m^2-2mt^\star+2(t^\star)^2)}\,\\
		L_2 &= \frac{m\rhoid(m)}{(m-t^\star)(m-t^\star+1)} = 
		\frac{m^2}{(m-t^\star+1)(m^2-2mt^\star+2(t^\star)^2)}\,,\\
		U_1 &= \frac{m(\rhoid(m)-1)}{t^\star(t^\star-1)} 
		=\frac{m(m-2t^\star)}{(t^\star-1)(m^2-2mt^\star+2(t^\star)^2)}\,,\\
		U_2 &= \frac{\rhoid(m)}{t^\star-1} = 
		\frac{m(m-t^\star)}{(t^\star-1)(m^2-2mt^\star+2(t^\star)^2)}\,,\\
		U_3 &= \frac{m\rhoid(m)}{(m-t^\star)(m-t^\star-1)} = 
		\frac{m^2}{(m-t^\star-1)(m^2-2mt^\star+2(t^\star)^2)}\,.
	\end{align*}
	First notice that $U_1$ and $U_2$ differ only by $mt^\star$ in the numerator,
	so $U_1$ is not greater than~$U_2$ and therefore we can ignore $U_2$. Moreover,
	when comparing $L_2$ with $U_3$ and $L_1$ with~$U_1$, we observe that the
	numerators do not differ while the denominators are smaller in the upper
	bounds. Therefore, it is immediate that $L_2\le U_3$ and $L_1\le U_1$.
	
	For the remaining two comparisons, $L_2\le U_1$ and $L_1\le U_3$, 
	we use that $t=t^\star$ maximizes the expression
	\[
		\compfunc(t) := \frac{m(m-t)}{m^2-2mt+2t^2} = \frac{m(m-t)}{(m-t)^2+t^2}\,.
	\]
	Hence,~$\compfunc(t^\star)-\compfunc(t^\star-1)\ge0$. Therefore, we have that 
	\begin{align}
		0&\le \frac{m(m-t^\star)}{(m-t^\star)^2+(t^\star)^2} - 
		\frac{m(m-(t^\star-1))}{(m-(t^\star-1))^2+(t^\star-1)^2}\nonumber\\
		& = m\frac{(m-t^\star)^2+(m-t^\star)(-2t^\star+2)
			-(t^\star)^2}
		{\left((m-t^\star)^2+(t^\star)^2\right)\left((m-t^\star+1)^2+(t^\star-1)^2\right)}\,.
		\label{eq:intermtstaropt-}
	\end{align}
	Since $m$ and the denominator of~\eqref{eq:intermtstaropt-}  are both strictly positive, we obtain that
	\begin{equation}
		m^2-4mt^\star+2m + 2(t^\star)^2-2t^\star\ge 0\,. \label{eq:tstaropt-}
	\end{equation}
	Similarly,~$g(t^\star)-g(t^\star+1)\ge0$ and we have that
	\begin{align}
		0&\le \frac{m(m-t^\star)}{(m-t^\star)^2+(t^\star)^2} - 
		\frac{m(m-(t^\star+1))}{(m-(t^\star+1))^2+(t^\star+1)^2}\nonumber\\
		& = m\frac{-(m-t^\star)^2+(m-t^\star)(2t^\star+2)+(t^\star)^2}
		{\left((m-t^\star)^2+(t^\star)^2\right)\left((m-t^\star-1)^2+(t^\star+1)^2\right)}\,.
		\label{eq:intermtstaropt+}
	\end{align}
	Again, $m$ and the denominator of~\eqref{eq:intermtstaropt+} are both strictly positive, so we obtain that
	\begin{equation}
		-m^2+4mt^\star+2m - 2(t^\star)^2-2t^\star\ge 0\,. \label{eq:tstaropt+}
	\end{equation}
	Now, we show that $L_2\le U_1$. This is equivalent to showing that
	\begin{align*}
	\frac{m}{m-t^\star+1}\le \frac{m-2t^\star}{t^\star-1}\quad 
	&\Leftrightarrow \quad 
	0\le m^2-4mt^\star+2m+2(t^\star)^2-2t^\star\,,
	\end{align*}
	which is true by \eqref{eq:tstaropt-}. 
	Finally, we show that $L_1\le U_3$. 
	This is equivalent to showing that
	\begin{align*}
	\frac{m-2t^\star}{t^\star+1}\le\frac{m}{m-t^\star-1}
	\quad & \Leftrightarrow \quad 
	m^2-4mt^\star-2m+2(t^\star)^2+2t^\star \le 0\,,
	\end{align*}
	which is true by \eqref{eq:tstaropt+}. 
	Thus, the interval $[\max\{L_1,L_2\},\min\{U_1,U_2,U_3\}]$ is non-empty 
		and there are valid choices for~$\delta$. 
	The validity of the upper bound follows.
\end{proof}	

\section{Proofs for \texorpdfstring{\Cref{sec:BricksOnSpeed}}{} -- Equal-Size Jobs and General Speeds}\label{apx:BricksOnSpeed}
	
	In the following we show the results of Lemma \ref{lem:speedTwoThreeMach} for few machines.
	\begin{manuallemma}{\ref{lem:speedTwoThreeMach}.a}
		The optimal algorithm for \srs for unit-size jobs has robustness factor~$4/3$ on~$m=2$ machines.
	\end{manuallemma}
	\begin{proof}
		The lower bound is implied by \Cref{thm:speedsandLB}.
		Let~$n$ be the number of jobs of the instance. 
		Consider an algorithm that builds two bags containing at most~$a_1$ and~$a_2$ jobs as follows.
		\[
		a_1 := \floor{\frac 43 \cdot \floor{n/4+1}}
		\quad ;\quad
		a_2:=\floor{\frac 43 \ceil{n/2}}\geq \floor{\frac 23 n}.
		\]
		
		We now show that (i) for every adversary, the algorithm can schedule these bags
		within a makespan of~$4/3$ and (ii) that~$a_1+a_2\geq n$, so the bags contain all
		jobs.
		
		The adversary places at least~$M_2:=\ceil{n/2}$ jobs on one machine, say
		machine~2. So the algorithm is always able to place at least~$a_2$ jobs on
		machine~2. If the adversary places at least \ceil{3n/4} jobs on machine 2,
		then the algorithm can place both bags on this machine. Otherwise, the
		adversary can place at most~$\ceil{3n/4}-1$ jobs on machine~2. This implies
		that there are at least~$M_1 := n-\ceil{3n/4}+1 = \floor{n/4}+1$ jobs
		on machine~1. This means that the algorithm can place at least~$\floor{\frac
			43 M_1}$ jobs on machine~1; this is exactly~$a_1$. So the makespan achieved by
		the algorithm is at most~$4/3 \cdot \cmax^*$, where~$\cmax^*$ is the optimal
		makespan.
		
		Hence, the algorithm is $4/3$-robust if~$a_1+a_2\geq n$. Consider the four
		natural integers $s<4$,~$t<3$,~$k$ and~$q$ such that~$n=4k+s=3q+t$. Note that
		\[
			a_1 = \floor{\frac 43 \cdot (k+1)} = \floor{\frac{4k+4}{3}} = \floor{\frac{3q+t + (4-s)}{3}} = q + \floor{\frac{4+t-s} 3}\,.
		\]
		We consider several cases which together complete the proof:
		\begin{itemize}
			\item~$4+t-s\geq 3$, i.e.,~$t\geq s-1$: we have~$a_1  \geq q+1\geq \ceil{n/3}$ so~$a_1+a_2\geq n$.
			\item~$t=0$: we have~$a_1 \geq q = \ceil{n/3}$ so~$a_1+a_2\geq n$.
			\item~$0<t<s-1\leq 2$, which means~$t=1$ and~$s=3$: we have~$a_1=q$ and~$n=4k+3$ is odd, so~$M_2 = \frac{1}{2}(n+1) = \frac{1}{2}(3q+2)$ and 
			\[
				a_2=\floor{\frac43 M_2}=\floor{2q+\frac43}=2q+1=n-a_1\,.
			\]
		\end{itemize}
	\end{proof}

	\begin{manuallemma}{\ref{lem:speedTwoThreeMach}.b}
		The optimal algorithm for \srs for unit-size jobs has robustness factor~$3/2$ on~$m=3$ machines.
	\end{manuallemma}
	\begin{proof}
		The lower bound is implied by \Cref{lem:LB3:1.5}.
		Let~$n>3$ be the number of jobs. Consider an algorithm that builds three bags as follows:
		\begin{align*}
		a_1 := \floor{\frac 32 \cdot{\frac12\floor{\frac n3+1}}}\quad;\quad
		a_3 := \floor{\frac 32\ceil {\frac n3}} \geq \floor{n/2} \geq \frac{n-1}{2}\quad ; \quad
		a_2 := n-a_1-a_3\,.
		\end{align*}
		
		Order the machines in increasing order of adversary load first. The adversary places at least
		$\ceil{\frac n3}$ jobs on the most loaded machine, machine 3, so the algorithm
		can always put at least~$a_3$ jobs, i.e., the third bag on machine 3.
		
		If the adversary places at least~$\ceil{\frac{2}{3}n}$ jobs on machine 3, then
		the algorithm can put all bags on it and the claim holds. Assume now the
		adversary places at most~$\ceil{\frac{2}{3}n}-1$ jobs on machine 3, so
		machines 1 and 2 receive at least~$\floor{\frac n3+1}$ jobs combined. This
		means in particular that the algorithm can always put bag~$a_1$ on machine 2.
		We now consider several cases that could prevent the algorithm from reaching a
		robustness factor of~$3/2$. They all implicitly assume that the algorithm
		cannot simultaneously put the bags~$a_1$ and~$a_2$ on machines 1 and 2 and
		cannot put the bags~$a_2$ and~$a_3$ jobs on machine 3, as the contrary allows
		to fit all bags. We therefore show a contradiction in each case.
		\setlist[enumerate,1]{label={(\roman*)}}
		\begin{enumerate}
			\item The algorithm cannot put~$a_2$ on machine 2. This means than~$M_2< \frac 23 a_2$ 
			and therefore changed from itemize (with dashes) to enumeration
			\begin{align*}
				M_1\leq M_2 &\leq \frac 23 a_2 - \frac 13\quad;\quad
				 M_3 \leq \frac 23 (a_2+a_3) -\frac 13\,.
			\end{align*}
			This implies
			\[
				n= M_1+M_2+M_3\leq 2a_2 + \frac 23 a_3 -1.
			\]
			Using that~$a_2 = n - a_1 - a_3$ and rearranging yields
			\[
				n \leq 2n - 2a_1 - 2a_3 + \frac 23 a_3 -1 \quad \Leftrightarrow\quad
				2a_1 + \frac 43 a_3-n+1\leq 0\,.
			\]
			However, letting~$n= 3k+t$ with~$t\in\{0,1,2\}$ and~$k>0$, we have 
			\[
				a_1= \floor{\frac{3k+3}{4}} \geq \frac{3k}{4} \quad ; \quad a_3 \geq \frac{n-1}2\,.
			\]
			This leads to the following contradiction:
			\begin{align*}
				2a_1 + \frac 43 a_3 - n + 1 & \geq \frac {3k}2 + \frac 23 n - \frac 23 - n + 1
				= \frac{3k}{2} - \frac n 3 + \frac 13 \geq k + \frac{k}{2} - k - \frac 13>0\,.
			\end{align*}
			
			\item The algorithm cannot put~$a_3$ on machine 2,~$a_2$ can be put there.
			Consequently, the algorithm cannot put~$a_1$ on machine 1 as this allows to
			place simultaneously~$a_1$ and~$a_2$ on machines 1 and 2. So~$M_1 < \frac 23
			a_1$, which means~$M_1 \leq \frac23 a_1 - \frac 13$. Similarly, the algorithm
			cannot put simultaneously~$a_1$ and~$a_2$ on machine 2 nor~$a_2$ and~$a_3$ on
			machine~3. Therefore, we have
			\begin{align*}
				M_1 &\leq \frac23 a_1 - \frac 13\quad;\quad
				M_2 \leq \frac 23 \min(a_1+a_2,a_3) - \frac 13\quad;\quad
				M_3 \leq \frac 23 (a_2+a_3) -\frac 13\,.
			\end{align*}
			
			Noting that~$\min(a_1+a_2,a_3)\leq n/2$ as~$a_1+a_2+a_3=n$, we get the following contradiction
			\begin{align*}
				n = M_1+M_2+M_3 &\leq \frac 23 (a_1+a_2+a_3 + \min(a_1+a_2,a_3)) -1 \\
				&\leq \frac 23 (n + \frac 12 n) -1= n-1\,.
			\end{align*}
			
			\item The algorithm cannot put~$a_1+a_2$ on machine 3.  By definition, we
			have~$a_3 \geq \frac{n-1}{2}$, so~$a_1+a_2\leq \frac{n+1}{2}$ and~$M_3 <
			\frac 23 (a_1+a_2)= \frac{n+1}{3}$. Therefore, all machine loads equal~$n/3$
			so one bag per machine fits, which is a contradiction.
		\end{enumerate}
	\end{proof}
	
	\begin{manuallemma}{\ref{lem:speedTwoThreeMach}.c}\label{lem:LBm6}
		For~$m=6$, the optimal algorithm for speed-robust scheduling for unit-size jobs
		has a robustness factor larger than~$\rhosand$. 		
	\end{manuallemma}
	\begin{proof}
	
		Consider~$n=756$ unit-size jobs and~$m=6$ machines. Consider any algorithm
		building~$6$ bags of sizes~$a_1\leq a_2\leq a_3\leq a_4\leq a_5 \leq a_6$ out
		of these jobs. Consider an adversary setting where five machines are set to a
		speed~$s_1\in \mathbb N$ and one machine to a speed~$s_6\in \mathbb N$ with
		$5s_1+s_6 \geq n$ such that an optimal schedule of~$n$ jobs on these machines
		has a makespan at most 1. For the algorithm and for each~$i$, either a bag of
		size at least~$a_i$ is scheduled on a machine of speed~$s_1$ or all bags of
		size at least~$a_i$ are scheduled on the machine of speed~$s_6$. Hence, for
		each~$i\in\{1,\dots,6\}$, the number~$\phi_i := \min(a_i/s_1, \sum_{j=i}^{6}
		a_j/s_6)$ is a lower bound on the algorithm's makespan, in other words, on its
		robustness factor.
		
		Consider \cref{tbl:LMm6}. Observe that the Conditions~$(2)$ and~$(3)$ are
		equivalent as~$n$, the number of jobs, is fixed. If none of the Conditions
		$(1)$ are satisfied, then the total size of the bags is at most
		$76+91+109+132+158+189=755<n$, which is a contradiction. Hence, let~$i$ be the
		first row such that Condition~$(1)$ is satisfied. The fact that
		Condition~$(1)$ is not satisfied for~$i'< i$ then implies that Condition~$(3)$
		is met by definition. Hence, the bag sizes of the algorithm satisfy the
		conditions of at least one row. Observe that this implies that for this
		particular~$i$,~$\phi_i$ is then a lower bound on the makespan of the
		algorithms assignment. Overall,~$\phi = \min_i \phi_i$ is then a lower bound
		on the robustness factor of the algorithm.
		
		We have that~$\phi = \frac{589}{391}\approx 1.506 > \rhosandm{6}\approx
		1.503$. Hence, the robustness factor of \emph{any} algorithm for unit-size
		jobs and~$m$ machines is indeed strictly larger than the robustness factor $\rhosand$ for
		infinitesimal jobs on the same number of machines; see \cref{thm:speedsandUB}.
	\end{proof}
	
	\begin{table}[tbh]
		\centering
		\renewcommand{\arraystretch}{1.4}
		\begin{tabular}{|c|c|c|c|c|c|}
			\hline
			\multicolumn{3}{|c|}{Conditions on bag sizes}  &  \multirow{2}{*}{$s_1$} & \multirow{2}{*}{$s_6$} & \multirow{2}{*}{$\phi$} \\ \cline{1-3}
			$(1)$ &~$(2)$ &~$(3)$ &&&  \\ \hline
			$a_1\geq 77$ & &
			&  51 & 501 &~$\min(\frac{77}{51},\frac{756}{501})$\\
			$a_2\geq92$ &~$\sum_{j\geq 2} a_j \geq 680$ &~$a_1 \leq 76$ 
			&  61 & 451 &~$\min(\frac{92}{61},\frac{680}{451})$\\
			$a_3\geq 110$ &~$\sum_{j\geq 3} a_j \geq 589$ &~$a_1+a_2 \leq 167$
			&  73 & 391 &~$\min(\frac{110}{73},\frac{589}{391})$\\
			$a_4\geq 133$ &~$\sum_{j\geq 4} a_j \geq 480$ &~$\sum_{j< 4} a_j \leq 276$
			&  88 & 318 &~$\min(\frac{133}{88},\frac{480}{318})$\\
			$a_5\geq 159$ &~$a_5+a_6 \geq 348$ &~$\sum_{j< 5} a_j \leq 408$
			&  105 & 231 &~$\min(\frac{159}{105},\frac{348}{231})$\\
			$a_6\geq 190$ & &
			&  126 & 226 &~$\frac{190}{126}$\\
			\hline
		\end{tabular}
		\caption{Speed instance in function of the bag sizes in \Cref{lem:LBm6}.}
		\label{tbl:LMm6}
	\end{table}

\section{Proofs for \texorpdfstring{\Cref{sec:BricksUB}}{} -- Equal-Size Jobs and Speeds in \texorpdfstring{$\{0,1\}$}{0,1}}
\label{apx:BricksUB}

	\renewcommand{\k}{\ensuremath{{\bar\lambda}}}

	In this section, we give the missing proof of \Cref{lem:bricks:3<=Om<=8}. For the remainder of this
	section, we use~$\k := \optm$ to denote the optimal makespan on~$\nrbags$
	machines. 
	\bricksforthreeOmeight*
	The proof of \Cref{lem:bricks:3<=Om<=8} consists of two major cases depending on the number
	of bags that \lpt assigns to the same machine. The first part of the proof is to
	consider~$\nrmach$ such that \lpt assigns at most~$2$ bags to any machine (\Cref{lem:bricks:m'>=m/2}), and
	the second part consists of~$\nrmach$ such that there is at least one machine to
	which \lpt assigns at least~$3$ bags (\Cref{lem:bricks:m'<m/2}).	
	
	Recall that for~$3 \leq \k \leq 8$, we pack four different types~$a_0,\ldots, a_3$
	of bags depending on~$\k$, where~$a_l$ denotes the size of the~$l$-th bag type and~$x_l$ its multiplicity. 
	We give their values in the following two tables. 	
	\begin{equation}
		\label{eq:Om<=8}\tag{$\mathcal{B}$}
		\begin{array}{c||rrrrrr}
			\k & 3 & 4 & 5 & 6 & 7 & 8 \\\hline
			a_0 & 1 & 2 & 3 & 3 & 4 & 5  \\
			a_1 & 2 & 3 & 4 & 4 & 5 & 6  \\
			a_2 & 3 & 4 & 5 & 6 & 7 & 8 \\
			a_3 & 4 & 5 & 6 & 8 & 9 & 10 
		\end{array} \qquad
		\renewcommand*\arraystretch{1.3}
		\begin{array}{c||ccccc|c}
			m\!\Mod 5 & 0 & 1 & 2 & 3 & 4 & a_l\\\hline
			x_0 \!+\! x_1 & \frac{2m}5 & 2 \floor{\frac m5} & 2 \floor{\frac m5} \!+\! 1 & 2 \floor{\frac m5} \!+\! 1 & 2 \floor{\frac m5} \!+\! 1 & \ceil{\frac23 \k}\\
			x_2 & \frac m5 & \floor{\frac m5} \!+\! 1 & \floor{\frac m5} & \floor{\frac m5} \!+\! 1 & \floor{\frac m5} \!+\! 2 & \k\\
			x_3  & \frac{2m}5 & 2 \floor{\frac m5} & 2 \floor{\frac m5} \!+\! 1 & 2 \floor{\frac m5} \!+\! 1 & 2 \floor{\frac m5} \!+\! 1 & \floor{\frac43\k}
		\end{array}
	\end{equation}

	\begin{lemma}\label{lem:bricks:m'>=m/2}
		Let $\lptf$ denote the makespan of \lpt that assigns bags as described 
		by~\eqref{eq:Om<=8} to $m'$ machines. 
		If~$\nrmach \geq \frac\nrbags2$, then~$\lptf \leq \frac43 \optf$.
	\end{lemma}
	\begin{proof}
		To prove this lemma, we consider all cases of how \lpt assigns two bags to the
		same machine. Let~$\lptf$ denote the resulting makespan.
		We start by bounding the number of failing machines~$\nrfail$ depending on the
		value of~$\lptf$. If~$\lptf = a_{l'} + a_l$ with~$l'\leq l$, then
		\begin{equation}\label{eq:bricks:mfailed}
			\nrfail \geq \sum_{l'' =0}^{l'-1} x_{l''} + \floor{\sum_{l'' = l'}^{l-1} x_{l''}/2} +1\,.
		\end{equation}
		\smallskip
		
		\noindent {$\bm{\text{LPT}_\nrmach\in  \{2a_0, a_0 + a_1, 2a_1\} }$:} 
		Since~$\k < \opt{\nrbags-1}$ by our assumption based on \cref{lem:m-1}, we have~$2 a_1 \!\leq\! \frac43 \optf$ if~$\nrmach \!<\! \nrbags$. So
		if~$\lptf \!\in\! \{2a_0, a_0 \!+\! a_1, 2a_1\}$, then~$\lptf \!\leq\! \frac43 \optf$. 
		
		\smallskip
		
		\noindent{$\bm{\text{LPT}_\nrmach = a_0 + a_2}$:} 
		Observe that~$\floor{(x_0+x_1)/2}\geq \floor{\frac{m}{5}}$ by definition, see~(\ref{eq:Om<=8}). 
		With Equation~\eqref{eq:bricks:mfailed}, we obtain~$\nrmach = \nrbags -\nrfail \leq m - (\lfloor m/5 \rfloor + 1) \leq \frac45 m$. This implies that
		\begin{equation*}
			\optf \geq \frac{n}{\nrmach} 
			\geq \frac{\nrbags \k - \ell}{4/5 \nrbags} 
			= \frac54\left(\k - \frac{\ell}{\nrbags}\right).
		\end{equation*}
		As~$\ell < \k \leq 8$ and~$m\geq 37$ by assumption, we have that~$\frac54 \frac\ell m \leq \frac14$. Hence, 
		\begin{equation*}
			\frac43 \optf  \geq \frac43 \left(\frac54\k - \frac14\right) \geq \frac53 \k - \frac13\,.
		\end{equation*}
		Using that~$a_0 + a_2 = \left(\floor{2/3\k}-1 \right) + \k \leq \frac53\k -
		1$, we conclude that~$\lptf \leq \frac43 \optf$.
		
		\medskip For the remaining cases, let~$l$ and~$l'$ be the indices of the bag
		types that are assigned to the same machine and let~$y_{ll'} := \sum_{l''
		=0}^{l'-1} x_{l''} + \floor{\sum_{l'' = l'}^{l-1} x_{l''}/2}$. Then,~$\nrfail
		\geq y_{ll'} + 1$ by Equation~\eqref{eq:bricks:mfailed}. Showing that~$\lptf
		\leq \frac43 \optf$ is equivalent to showing that
		\[
			a_{l'} + a_l \leq \frac43 \ceil{\frac{n}{\nrbags - (y_{ll'}+1)}}
		\]
		for all possible combinations of~$a_{l'}$ and~$a_{l}$. 
		With~$a_{ll'} := \ceil{\frac34(a_{l'} + a_{l})}$, this inequality holds if 
		\begin{equation}\label{eq:bricks:y>=...}
			y_{ll'} \geq \floor{\frac{(a_{ll'}-1)\nrbags - n}{a_{ll'}-1}} = \nrbags -
			\ceil{\frac{n}{a_{ll'}-1}} = \floor{\nrbags - \frac{\k m}{a_{ll'}-1} +
			\frac{\ell}{a_{ll'}-1}}\,,
		\end{equation}
		where we used the fact that~$y_{ll'} \in \mathbb{Z}$. 
		
		\smallskip
		
		\noindent{$\bm{\text{LPT}_\nrmach = a_1 + a_2}$:} By Equation~\eqref{eq:bricks:y>=...}, 
		it suffices to verify~$x_0 + \floor{\frac{x_1}{2}} \geq \nrbags - \ceil{\frac{n}{\ceil{3/4(a_1 + a_2)}-1}}$. 
		If~$\k = 3$, then the right hand side becomes~$m - \ceil{\frac{3\nrbags - \ell}{3}} = 0$. 
		Hence, the inequality is satisfied. For~$\k \geq 4$, the right hand side is at 
		most~$\floor{\frac \nrbags5 + \frac\ell5} \leq \floor{\frac \nrbags5} + \floor{\frac45 + \frac{\ell}{a-1}}$ 
		with~$0 \leq \ell < \k$. Observe that~$\floor{\frac{x_1}2} + x_0 \geq 
		\floor{\frac \nrbags5} + \floor{\frac{x_0}2}$. Using the definition of~$x_0 + x_1$ and~$x_0 = \ell$, 
		\begin{equation*}
			\floor{\frac \nrbags5} + \floor{\frac45 + \frac{\ell}{a-1}} \leq 
			\begin{cases}
				\floor{\frac \nrbags5} + 0 \leq \floor{\frac{x_1}{2}} + x_0 & \text{ if } x_0 = 0 \\
				\floor{\frac \nrbags5} + 0 \leq \floor{\frac{x_1}{2}} + x_0 & \text{ if } x_0 = 1 \text{ and } \k \geq 5\\
				\floor{\frac \nrbags5} + 1 \leq \floor{\frac{x_1}{2}} + x_0 & \text{ if } x_0 \geq 2\,,
			\end{cases}	
		\end{equation*}
		which shows the validity of Inequality~\eqref{eq:bricks:y>=...} in all cases
		except the combination of~$\ell = x_0 = 1$ and~$\k =4$. For this particular
		case, a careful case distinction based on~$\nrbags \Mod 5$ shows that
		Inequality~\eqref{eq:bricks:y>=...} still holds.

		\noindent{$\bm{\text{LPT}_\nrmach= a_2 + a_2}$:} 
		By Equation~\eqref{eq:bricks:y>=...}, it suffices to verify~$x_0 + x_1 \geq
		\nrbags - \ceil{\frac{n}{\ceil{3a_2/2}-1}}$. The right hand side can be
		transformed into~$\floor{\frac \nrbags3 + \frac23\frac\ell \nrbags} \leq 
		\frac{\nrbags}{3} + \frac23 = \frac5{15}\nrbags + \frac{10}{15}$. We have~$x_0
		+ x_1 \geq 2 \floor{\frac \nrbags 5} \geq \frac{6}{15} \nrbags-
		\frac{24}{15}$. Using that~$\nrbags \geq 50$, we obtain~$x_0 + x_1 \geq
		\frac{5}{15}\nrbags + \frac{21}{15}$, which concludes the proof of
		Inequality~\eqref{eq:bricks:y>=...} for this case.
		
		\smallskip
		
		\noindent{$\bm{\text{LPT}_\nrmach = a_0 + a_3}$:} We need to verify 
		\begin{align*}
			\frac{x_0 + x_1 + x_2}2 & \geq 
			\floor{\nrbags - \frac{\k}{\floor{3\k/2} -1}\nrbags + \frac\ell{\floor{3\k/2} -1}}\,.
		\end{align*} 
		Based on~$\nrbags \Mod 7$, the last term can be bounded from above by 
		\begin{align*}
			\floor{ \frac27 \nrbags + \frac23}
= \frac27 \nrbags - 2\frac{\nrbags \Mod 7}{7} + \floor{2\frac{\nrbags \Mod 7}{7} + \frac23}
\leq \frac27\nrbags + \frac{40}{70}\,.
		\end{align*}
		Consider~$\frac{x_0 + x_1 + x_2}{2}$, the left hand side of the inequality. We can express this as 
		\[
		\frac{x_0 + x_1 + x_2}{2} = \frac{3}{10}\nrbags + R
		\geq \frac27\nrbags + \frac{43}{70}\,,
		\]
		where $R \in \{0, \frac2{10}, -\frac1{10}, \frac1{10}, \frac3{10}\}$ and we used that~$\nrbags \geq 50$. 
		This shows the validity of Inequality~\eqref{eq:bricks:y>=...} for the current case. 
		
		\smallskip
		
		\noindent{$\bm{\text{LPT}_\nrmach = a_1 + a_3}$:} Since~$a_1 + a_3 = 2 \k$, we need to verify 
		\[
			x_0 + \frac{x_1 + x_2}{2} \geq \floor{\nrbags - \frac{\k}{\ceil{3\k/2} -1}  + \frac{\ell}{\ceil{3\k/2} -1}}\,.
		\]
		Using~$\k \leq 8$ and that the
		second term on the right hand side is increasing in~$\k$ and depends on the
		parity of~$\k$, we can upper bound the right hand side
		by~$\floor{\frac{3}{10}\nrbags + \frac{2}{3}\frac{\ell}{\k-1}}$. Note that the
		left hand side is slightly larger than~$\frac3{10}m$ but not sufficiently large
		for a crude upper bound. Hence, we rewrite the left hand side as
		\begin{equation*}
			x_0 + \frac{x_1 + x_2}{2} = \frac3{10}\nrbags + \frac{x_0}{2} + R,
		\end{equation*}
		with $R\in \{0, \frac2{10}, -\frac1{10}, \frac1{10},\frac3{10} \}$ depending on~$\nrbags \Mod 5$. 
		
		If~$x_0 = \ell =0$, then the right hand side of the inequality
		is~$\floor{\frac3{10}m}$. Except for the case~$\nrbags \Mod 5 = 2$, the term~$\frac{x_1 +
		x_2}{2}$ clearly satisfies the inequality. If~$\nrbags \Mod 5 =2$, we
		have~$\nrbags \Mod{10} \in \{2,7\}$, which implies that~$\floor{\frac3{10}m}
		\leq \frac3{10}m - \frac{6}{10} \leq \frac{x_1 + x_2}{2}$ by the case
		distinction above.
		
		If~$x_0 = \ell > 0$, then~$\frac{x_0}2 \geq \frac12$ and, thus, the case distinction yields 
	$
			x_0 + \frac{x_1 + x_2}{2} \geq \frac3{10}\nrbags + \frac4{10}\,.
	$
		Using that the right hand side is upper bounded by~$\floor{\frac3{10}\nrbags +
		\frac23}$, we use a similar case distinction based on~$\nrbags \Mod {10}$ to
		derive
		\begin{align*}
			\floor{\frac3{10}\nrbags + \frac23} = \frac3{10}\nrbags - \frac3{10}(\nrbags \Mod {10}) 
			+ \floor{\frac{3}{10} (\nrbags \Mod {10}) + \frac23} 
			 \leq \frac{3}{10}\nrbags + \frac{3}{10}\,,
		\end{align*}
		which concludes the proof of Inequality~\eqref{eq:bricks:y>=...}. 
		
		\smallskip By our choice of~$x_0,\ldots,x_3$, the occurrence of~$\lptf \in \{ a_2
		+ a_3, a_3 + a_3\}$ implies that~$\nrmach < \frac{\nrbags}{2}$ which is not
		considered in this lemma.
	\end{proof}
	
	\begin{lemma}\label{lem:bricks:m'<m/2}
		Let $\lptf$ denote the makespan of \lpt that assigns bags as described 
		by~\eqref{eq:Om<=8} to $m'$ machines. 
		If~$\nrmach < \frac\nrbags2$, then~$\lptf \leq \frac43 \optf$.
	\end{lemma}
	
	\begin{proof}	
		Assume for the sake of contradiction that \lpt fails to place all bags onto
		the machines such that~$\lptf \leq \frac 43 \optf$. We consider the first
		bag~$b\in [\nrbags]$ whose assignment to the currently least loaded machine
		causes the failure, i.e., the completion time of this machine exceeds~$\frac43
		\optf$. For simplicity, let this be machine~$i$ and let~$C_i$ be the completion
		time of~$i$ \emph{before} adding bag~$b$. Let~$a$ be the size of bag~$b$.
		
		As~$\nrmach < \frac\nrbags2$,~$\nrbags \geq 37$, and~$\k < \opt{\nrbags-1}$, we
		have that~$\optf \geq 2\k + 1$. Hence, if bag~$b$ is the first or second bag on
		machine~$i$, then~$C_i + a \leq 2 \floor{\frac43 \k} \leq \frac43 \optf$; a
		contradiction.
		
		Consider the case where bag~$b$ is the fourth bag (or larger) on machine~$i$
		and restrict the instance to consist only of the jobs assigned by \lpt so
		far plus the jobs in bag~$b$. Let~$\optf'$ be the optimum of this restricted
		instance on~$\nrmach$ machines. As \lpt has not assigned all bags yet, i.e.,
		there is still unscheduled volume, we have that~$\optf' \geq C_i + 1$. Since
		there are already at least three bags of size at least~$a$ on machine~$i$, we
		have~$\optf' \geq C_i \geq 3 a$. Moreover, as~$b$ is the first bag to
		violate~$C_i + a \leq \frac43 \optf$, bag~$b$ determines the makespan of \lpt
		on the restricted instance; a contradiction by \cref{lem:bricks:lpt}.
		
		Hence, bag~$b$ is the third bag on machine~$i$. Based on its size, 
		we distinguish four cases.

		\smallskip
		
		\noindent{$\bm{a = a_0}$:} If a bag of size~$a_0$ is the first bag whose
		completion time violates~$\frac43 \optf$, then a bag of size~$a_0$
		determines~$\lptf$. Note that~$3 a_0 = 3 (\ceil{\frac23\k} - 1) \leq 3\frac23\k
		= 2\k \leq \optf$. Hence, by \cref{lem:bricks:lpt}, we have~$\lptf \leq \frac43
		\optf$; a contradiction.
		
		\smallskip
		
		\noindent{$\bm{a = a_1}$:} If~$C_i = 2 a_3$, then~$\frac43 \optf \geq
		\frac{10}{3}\k$ while~$C_i + a_1 \leq \frac{10}{3}\k$; a contradiction. If~$C_i
		= a_2 + a_3$, then~$\frac43 \optf \geq \frac{28}9 \k$, while~$C_i + a_1 = 3
		\k$; a contradiction. If~$C_i = 2\k = 2a_2 = a_1 + a_3$, then~$\frac43 \optf
		\geq \frac83\k + \frac43$ while~$C_i + a_1 \leq \frac83 \k + \frac23$; a
		contradiction. If~$C_i = a_1 + a_2$, then~$\frac43 \optf \geq a_2 + a_1 + (a_2
		+ a_1)/3 + \frac43$ while~$C_i + a_1 = a_2 + 2 a_1$. For~$3 \leq \k \leq 8$,
		one can check that~$(a_2 + a_1)/3 + \frac43 \geq a_1$; a contradiction.
		As~$\nrmach < \frac\nrbags2$, we have covered all possibilities for~$C_i$.
		Hence, a bag of size~$a_1$ cannot cause \lpt to fail.
		
		\smallskip
		
		\noindent{$\bm{a = a_2}$:} Let~$x$ be the number of bags of size~$a_2$ that
		were successfully assigned by \lpt before bag~$b$. Denote by~$V$ a volume of 
		$m' \cdot \frac 4 3 \optf$ minus the volume of already assigned bags
		which are~$x_3$ bags of size~$a_3$ and~$x$ bags of size~$a_2$. We have 
		\begin{align*}
		V & \geq \frac43 \left( a_0x_0 + a_1x_1 + (x_2 - x) a_2   \right) + \frac13 \left( xa_2 + a_3 x_3 \right)\,.
		\shortintertext{Using~$x_0 + x_1 = x_3$, this implies}
		V & \geq a_1 x_3 - \frac43 x_0 + \frac13 x_3\left(a_1  + a_3 \right) + \frac13 x_2a_2 + (x_2 - x) a_2 \geq \frac43 a_2 x_3\,,		
		\end{align*}
		where the second inequality follows from~$\frac43 x_0 < \frac43 \k \leq \frac13 x_2\k + (x_2 - x) \k$. As~$b$ is the third bag on machine~$i$ by the above discussion, we have~$\nrmach < \frac{x_3 + x_2}{2} \leq x_3$. Thus, we conclude that $V \geq \frac43 a_2$.
		Hence, the total volume left on the~$\nrmach$ machines is at least~$\nrmach
		a_2$. Hence, there has to be one machine~$i'$ where~$b$ still fits,
		i.e.,~$C_{i'} + a_2 \leq \frac43 \optf$. As~$i$ is the least loaded machine
		when \lpt assigns bag~$b$ and~$b$ violates~$\frac43 \optf$, we obtain a
		contradiction.
		
		\smallskip
		
		\noindent{$\bm{a = a_3}$:} Let~$x$ be again the number of bags of
		size~$a_3$ successfully assigned to machines by \lpt before bag~$b$. If~$a_3$
		is the size of a third bag on machine~$i$, then~$\nrmach < \frac{x_3}2$. By
		definition of~$x_2$, this additionally implies that~$\nrmach \leq x_2$. Let~$V$
		be the remaining volume after having assigned~$x$ bags of size~$a_3$. Then,
		\begin{align*}
			V & = \frac43 \left( a_0x_0 + a_1 x_1 + a_2 x_2 + a_3 (x_3 -x)\right) + \frac13 a_3 x
			 \geq \frac83 a_1 \nrmach + \frac43 a_2 \nrmach + \frac23 a_3 \nrmach\,,	  
		\end{align*}
		where we used~$\frac43 x_0 \leq a_3 \leq (x_3 - x) a_3$ 
		and~$\nrmach < \frac{x_3}2$ as well as~$\nrmach \leq x_2$.
		Thus, the remaining volume satisfies~$V \geq a_3 \nrmach$. Hence, there is at
		least one machine with remaining volume at least~$a_3$ contradicting \lpt's
		choice of machine~$i$.
		
		\smallskip As all possible cases for the size of the first bag that causes \lpt
		to fail lead to a contradiction, this proves the statement.
	\end{proof}
	
	\begin{proof}[\cref{lem:bricks:3<=Om<=8}]
		Combining the results of \cref{lem:bricks:m'>=m/2,lem:bricks:m'<m/2} shows that
		packing the bags according to~\eqref{eq:Om<=8} is $\frac43$-robust.
	\end{proof}

\section{Tradeoff Version of \Cref{thm:speedLPT} for Packing More Bags}
\label{app:extensions}

For the sake of completeness, we show an easy way how \Cref{thm:speedLPT} can be extended to the model in which we may partition the set of jobs into $M\geq m$ bags. 

	\begin{theorem}
		For~$M\geq m \ge1$, let $\alpha:= 1+\frac{m-1}{M}$. \lpt is~$\alpha$-robust for \srs.
	\end{theorem}
	\begin{proof}
	The proof follows the same lines as for \Cref{thm:speedLPT} with the adjustment that \lpt packs $M$ bags instead of $m$ and, for assigning bags to machines, we use a parameterized capacity bound~$\alpha \cdot s_i$ for machine~$i$ for some $\alpha \geq 1$.
	
	We slightly adjust the last part of the proof, where we give a lower bound on the total remaining capacity on
		the~$m$ machines when the second-stage algorithm fails to place the~$(k+1)$\nobreakdash-st bag. 
		The~$(M-k)$ bags that were not placed have a combined volume of at
		least~$V_\ell = (M-k-1)w+T \geq (M-k+1)\frac{T}{2}~$. The bags that were placed
		have a combined volume of at least~$V_p = kT$. The remaining capacity is then
		at least~$C =\alpha \cdot V_\ell + (\alpha-1)\cdot V_p$, and we have
		\begin{align*}
			C & \geq  \alpha \cdot(M-k+1) \frac{T}{2}+ (\alpha-1)\cdot kT 
			 = \frac{\alpha}{2} \cdot T M + T \left(\frac{\alpha k}{2} - \frac{\alpha k}{2} \right) +\frac{\alpha}{2}\cdot T -kT\\
			 & = \frac{\alpha}{2} \cdot T M + T \left(\frac{\alpha k}{2} + \frac{\alpha}{2} -k \right) 
			 = \frac{\alpha}{2} \cdot T M + T \left(\frac{\alpha}{2} -k\cdot  \left(1-\frac{\alpha}{2}\right)\right).
		\end{align*}	
		Using $\alpha<2$ and $k\leq M-1$, we obtain 
		\begin{align*}
			C & \geq \frac{\alpha}{2} \cdot T M + T \left(\frac{\alpha}{2} -(M-1)\cdot  \left(1-\frac{\alpha}{2}\right)\right).
		\end{align*}
		We define $x := M-m$ and replace $M$ by $m+x$. 
		Then
	\begin{align*}
			C & \geq T m - \left(1-\frac{\alpha}{2}\right) \cdot T m + \frac{\alpha}{2}\cdot Tx + T \left(\frac{\alpha}{2} -(m+x-1)\cdot  \left(1-\frac{\alpha}{2}\right)\right)\\
			 & \geq T m + T\cdot \left( -\left(1-\frac{\alpha}{2}\right)\cdot m + \frac{\alpha}{2}\cdot  x + \frac{\alpha}{2} - (m+x-1) \cdot  \left(1-\frac{\alpha}{2}\right) \right)\\
			 & = T m + T\cdot\left(\alpha\cdot(m+x) +1-2m-x\right) = Tm.
		\end{align*}
	Notice that the last equality follows directly from $\alpha=1+\frac{m-1}{M}$. 

		Thus, there is a machine with remaining capacity~$T$, which contradicts the 
		assumption that the bag of size~$T$ does not fit.
	\end{proof}

\end{document}